\documentclass[11pt]{article}
\usepackage{bbm}
\usepackage{url}
\usepackage{pstricks}
\usepackage{yfonts}
\usepackage{stmaryrd}
\usepackage{graphicx}
\usepackage{comment}
\usepackage{tipa}
\usepackage{mathrsfs}
\usepackage{algorithmic}
\usepackage{savesym}
\usepackage{amsmath,amsthm,amsfonts,amssymb}
\savesymbol{iint}
\restoresymbol{TXF}{iint}
\usepackage[letterpaper,margin=1.2in]{geometry}

\theoremstyle{plain}
\newtheorem{theorem}{Theorem}
\newtheorem{corollary}{Corollary}
\newtheorem{lemma}{Lemma}
\newtheorem{proposition}{Proposition}
\theoremstyle{definition}
\newtheorem{defn}{Definition}

\newtheorem{claim}{Claim}

\theoremstyle{remark}

\newtheorem{remark}{Remark}
\newtheorem{example}{Example}
\newcommand{\m}[1]{\xi^{*}(#1,{L,U})}
\newcommand{\mlu}[3]{\xi^*(#1,#2,#3)}
\newcommand{\bv}{b}
\newcommand{\gsp}{group-stra\-te\-gy\-proof }
\newcommand{\bpv}{b'} 
\newcommand{\A}{\mathcal{A}}
\newcommand{\ang}[1]{}
\newcommand{\man}[1]{}
\newcommand{\lcm}{Fence Monotonicity }
\newcommand{\lcmd}{Fence Monotonicity. }
\newcommand{\lcmc}{Fence Monotonicity, }
\usepackage{url}

\title{A complete characterization of group-strategy\-proof   mechanisms of cost-sharing}

\author{Emmanouil Pountourakis\thanks{University of Athens, Athens, Greece. Email: \texttt{e.pountourakis@di.uoa.gr}
} \and Angelina Vidali\thanks{Max-Planck-Institut f\"{u}r Informatik, Saarbr\"{u}cken, Germany, Email: \texttt{angelina@mpi-inf.mpg.de}}}
\begin{document}

\maketitle
\begin{abstract}
We study the problem of designing group-strategy\-proof cost-sharing mechanisms. The players report their bids for getting serviced and the mechanism decides which players are going to be serviced and how much each one of them is going to pay. 
We determine three conditions: \emph{Fence Monotonicity}, \emph{Stability} of the allocation and \emph{Validity} of the tie-breaking rule that are necessary and sufficient for group-strategy\-proofness, regardless of the cost function. Fence Monotonicity  puts restrictions only on the payments of the mechanism and stability only on the allocation. Consequently Fence Monotonicity characterizes group-strategy\-proof cost-sharing schemes.
Finally, we use our results to prove that there exist families of cost functions, where any group-strategy\-proof mechanism has unbounded approximation ratio.
\end{abstract}
\section{Introduction}
Algorithmic Mechanism Design~\cite{NR99} is a field of Game Theory, that tries to construct algorithms for allocating resources, that give to the players incentives to report their true interest in receiving a good, a service, or in participating in a given collective activity. The pivotal constraint when designing a mechanism for any problem is that it is truthful. Truthfulness also known as strategy-proofness or incentive compatibility requires that no player can strictly improve her utility by lying, when the values of the other players are fixed. In many settings this single requirement for an algorithm to be truthful restricts the repertoire of possible algorithms dramatically~\cite{Rob79}.

In settings where the repertoire of possible algorithms is not restricted too much by truthfulness, like for example in Cost-sharing problems it is desirable to construct mechanism that are also resistant to manipulation by groups of players. Group-stra\-te\-gy\-proofness naturally generalizes truthfulness by requiring that no group of players can improve their utility by lying, when the values of the other players are fixed. To be more precise there should not exist any group of players who can change their bids in a way that every member of the coalition is at least as happy as in the truthful scenario, and at least one person is happier, for fixed values of the players that do not belong to the coalition.

In this paper we study the following problem: We want to determine a set of $n$ customers/players, who are going to receive a service. Each player reports her willingness to pay for getting serviced and the mechanism decides which players are going to be serviced and the price that each one of them will pay, that is we consider direct revelation mechanisms. We want to characterize all possible mechanisms that satisfy group-strategy\-proofness, we want to find some necessary and sufficient conditions for a mechanism to be strategyproof and also to determine the corresponding payments.

We provide a complete characterization of \gsp mechanisms and cost-sharing schemes, closing a open question posed by Immorlica, Mahdian and Mirrokni~\cite{imm08,agt07} by extending the condition of Semi-cross-montonicity they identified in~\cite{imm08} to a new condition, which we call \lcm and by defining Fencing Mechanisms, a new general framework for designing \gsp mechanisms.

Our results are of special importance for a very important problem, the Cost-Sharing Problem, whose study was initiated in~\cite{m99}, where we additionally have a cost function $C$, such that for each subset of players $S$ the cost for providing service to all the players in $S$ is $C(S)$, however the strength of our results is that they apply for any cost function, since throughout our proof we do not make any assumptions at all about this cost function. We believe that our work here can be the starting point for constructing new interesting classes of mechanisms for specific cost-sharing problems.

Recently Mehta, Roughgarden and Sundararajan~\cite{mrs07} proposed the notion of weak group-strategy-proofness, that relaxes group-strategy\-proofness. It regards a formation of a coalition, as successful, when \emph{each} player who participates in the coalition \emph{strictly} increases her personal utility.  They also introduce acyclic mechanisms, a general framework for designing weakly group-strategy\-proof mechanisms, however the question of determining all possible weakly group-strategy\-proof mechanisms is an important question that remains open. Another alternative notion that is slightly stronger than weak-group-strategy\-proofness and weaker than group-strategy\-proofness was proposed by Bleischwitz, Monien and Schoppmann in~\cite{bms07}.

\section{Our results and related work}
The design of group-stra\-te\-gy\-proof mechanisms for cost-sharing was first discussed by Moulin and Shenker~\cite{m99,ms01}. Moulin defined a condition on the payments called cross-monotonicity, which states that the payment of a serviced player should decrease as the set of serviced players grows. Any mechanism whose payments satisfy cross-monotonicity can be easily turned to a simple mechanism called after Moulin. A Moulin mechanism first checks if all players can be serviced with positive utility and gradually diminishes the set of players that are candidates to be serviced, by throwing away at each step a player that cannot pay to get serviced (and who because of cross-monotonicity also cannot pay in any smaller set of serviced players). In fact if the cost function is sub-modular and 1-budget balanced
then the only possible group-stra\-te\-gy\-proof mechanisms are Moulin mechanisms~\cite{ms01}. The great majority of cost-sharing mechanisms proposed are Moulin mechanisms~\cite{jv01, pt03, kls05, gklrs07}. However recent results showed that for several important cost-sharing games Moulin mechanisms can only achieve a very bad budget balance factor~\cite{imm08,bs07,rs06,klsz05}. Another direction proposed by Moulin and rediscovered in~\cite{bs09,j08} resulting to weakly-group-stra\-te\-gy\-proof mechanisms are Incremental mechanisms where after ordering the players appropriately you ask them one by one if their bid is greater than an appropriate cost-share. Some alternative, very interesting and much more complicated in their description mechanisms that are group-stra\-te\-gy\-proof but not Moulin have been proposed in~\cite{imm08,pv06}, however these do not exhaust the class of group-stra\-te\-gy\-proof mechanisms.

In this work we introduce \emph{Fencing Mechanisms}, a new general framework for designing group-stra\-te\-gy\-proof mechanisms, that generalizes Moulin mechanisms~\cite{m99}. Unfortunately for the general case we do not know if there exists a polynomial-time algorithm that implements these mechanisms.

Finding a complete characterization of the cost-sharing schemes that give rise to  a group-stra\-te\-gy\-pr\-oof  mechanism was a question that was posed in~\cite{imm08,agt07}. The same question was posed in~\cite{mrs07} for weak group-stra\-te\-gy\-proof cost-sharing schemes and still remains open. Many interesting results arose in the attempt to find such a characterization~\cite{pv06,j08}. In contrast to previous characterization attempts that characterized mechanisms satisfying some additional boundary constraints~\cite{imm08,j08} our characterization is complete and succinct. The only complete characterization that was known was for the case of two players~\cite{pv06,j08}. It remains open how can our characterization help for constructing new efficient mechanisms for specific cost-sharing problems or for obtaining lower bounds and we belive that it can significantly enrich the repertoire of mechanisms with good approximation guarantees for specific problems.

In the notion of group-stra\-te\-gy\-proofness it is important to understand that ties play a very important role. This is in contrast to mechanisms that are only required to satisfy strategyproofness, where ties can be in most cases broken arbitrarily (see for example~\cite{nr01}).  An intuitive way to understand this is that a mechanism designer of a group-stra\-te\-gy\-proof
mechanism expects a player to tell a lie in order to help the other players increase their utility, even when she would not gain any profit for herself. This player is at a tie but decides strategically if she should lie or not. Consequently a characterization that assumes a priori a tie-breaking rule, and thus greatly restricts the repertoire of possible mechanisms, like the one in~\cite{imm08,j09} might be useful for specific problems and easier in its statement, but can never capture the very notion group-stra\-te\-gy\-proofness. 

We determine three conditions: \emph{\lcmc} \emph{Stability} of the allocation and \emph{Validity} of the tie-breaking rule that are necessary and sufficient for group-strategy\-proof\-ness, regardless of the cost function and without any additional constraints (like tie-breaking rules used in~\cite{imm08,j08}). \lcm  concerns only the payments of the mechanism, while Stability and Validity of the tie-breaking rule, only the allocation. Consequently \lcm  characterizes cross-monotonic cost-sharing schemes. Having only the payments of a \gsp mechanism is however not enough to determine its allocation. The allocation of a mechanism based on a  cost-sharing scheme that satisfies \lcmc should additional satisfy a condition we call Stability.  Managing to separate the payments from the allocation part of the mechanism and avoiding to add any additional restrictions  in the characterization we propose are undoubtedly its great virtues.

Our proofs are involved and based on set-theoretic arguments and the repeated use of induction. The main difficulty of our work was to identify some necessary and sufficient conditions for group-stra\-te\-gy\-proof payments that are also succinct to describe and add to our understanding of the notion of group-stra\-te\-gy\-proofness. In proving that \lcm is a necessary condition for group-strategyproofness we first have to prove Lemmas that also reveal interesting properties of the allocation part of the mechanism. A novel tool that we introduce is the \emph{harm relation} that generalizes the notion of negative elements defined in~\cite{imm08}.  Proving that Fencing Mechanisms, i.e. mechanisms whose payments satisfy \lcm and whose allocation satisfies Stability and Validity of the tie-breaking rule, are group-stra\-te\-gy\-proof turns out to be rather complicated.

\section{Defining the model}
\subsection{The Mechanism}
Suppose that $\A=\{1,2,\ldots ,n\}$ is a set of players interested in receiving a service. Each of the players has a private type $v_i$, which is her valuation for receiving the service.
\begin{defn}
A cost sharing mechanism $(O,p)$ consists of a pair of functions, $O:\mathbb{R}^n\rightarrow 2^{A}$ that associates with each bid vector $\bv$ the set of serviced players and $p:\mathbb{R}^n\rightarrow \mathbb{R}^n$ that associates with each bid vector $\bv$ a vector $p(\bv)=p(b_1,\ldots,b_n)$, where the $i$-th coordinate is the payment of player $i$.
\end{defn}

Each player wants to maximize her utility, which assuming quasi-linear utilities is $v_i a_i-p_i(\bv)$ where $a_i=1$ if $i\in O(b)$ and $a_i=0$ if $i\notin O(b)$.

As is common in the literature and in order not to come up with useless mechanisms, we concentrate on mechanisms that satisfy the following very simple conditions~\cite{m99}:
\begin{itemize}
\item \emph{Voluntary Participation (VP)}: A player that is not serviced is not charged ($i\notin O(\bv)\Rightarrow p_i(\bv)=0$) and a serviced player is never charged more than his bid ($i\in O(\bv)\Rightarrow p_i(\bv)\leq b_i$).
\item \emph{No Positive Transfer (NPT)}: The payment of each player $i$ is positive ($p_i(\bv)\geq 0$ for all $i$).
\item \emph{Consumer Sovereignty (CS)}: For each player $i$ there exists a value $b_i^*\in \mathbb{R}$ such that if he bids $b_i^*$ then it is guaranteed that $i$ will receive the service no matter what the other players bid.
\end{itemize}

Obviously, VP implies that if an player is truthful, then her utility is lower bounded by zero. Moreover, VP and NPT imply that if a player announces a negative amount, then she will not be included in the outcome. While, negative bids are not realistic, the latter may be used to model, the denial of revealing
any information to the mechanism. Finally, notice that the crucial value $\bv_i^*$, in the definition of CS is independent of the bid vector. Thus, this value has to be greater or equal to any possible payment this player is charged, i.e. $b^*_i> p_i(\bv)$, for all $\bv\in\mathbb{R}^n$.

\begin{defn} We say that a cost-sharing mechanism is
\emph{group-strategyproof (GSP)} if and only if the following holds:
For every two valuation vectors $v, v'$ and every $S \subseteq A$, sat\-
isfying $v_i = v'_i$ for all $i \in S$, one of the following is true:

\emph{(a)} There is some $i \in S$, such that $v'_i a_i - p_i (v' ) < v_i a_i -
p_i (v)$, or

\emph{(b)} for all $i \in S$, it holds that $v_i a'_i - p_i (v' ) = v_i a_i - p_i (v)$.
\end{defn}

In other words, a GSP mechanism does not allow success\-ful coalitions of the players, i.e. that a group of a players announces a false value, instead of their true valuations and
moreover no ``liar" sacrifices her utility, while at least one
player (not necessarily a liar) strictly profits after the manipulation.

\begin{defn}
A cost-sharing scheme is a function $\xi :\A\times 2^{\A}\rightarrow \mathbb{R}^+\cup \{0\}$, such that, for every $S\subset \A$ and every $i\notin S$ we have $\xi(i,S)=0$.
\end{defn}
You can think of $\xi(i,S)$ as the payment of player $i$ if the serviced set is $S$. In fact it can be shown that in any group-stra\-te\-gy\-proof mechanism for our setting the payment of a player depends only on the allocation of the mechanism and not directly on the bids of the players. In this sense we do not restrict the mechanism in any way by assuming that the payments are given by a cost-sharing scheme $\xi$.

\subsection{The cost function and budget balance}
The cost of providing service service is given by a cost function $C:2^{\A}\rightarrow \mathbb{R}^+\cup \{0\}$, where $C(S)$ specifies the cost of providing service to all players in $S$.

A desirable property of cost-sharing mechanisms is budget balance. We say that
a mechanism is \emph{$\alpha$-budget balanced}, where $0\leq \alpha\leq 1$, if for all bid vectors $b$ it holds, that $ \alpha \cdot C(O(\bv))\leq \sum_{i\in S}\xi(i,S)\leq C(O(\bv))$.

We chose to define the cost function last in order to stress that our results are completely independent of the cost function and apply to any cost-sharing problem.

\subsection{Cross- and Semi-cross-mo\-no\-to\-ni\-ci\-ty}

We say that a cost-sharing scheme is \emph{cross-monotonic}~\cite{ms01} if $\xi(i,S)\geq \xi(i,T)$ for every $S\subset T\subseteq \A$ and every player $i\in S$.

In the attempt to provide a characterization of GSP cost-sharing schemes Immorlica Mahdian and Mirrokni~\cite{imm08} provided a partial characterization and identified semi-cross monotonicity an important condition that should be satisfied by any GSP cost-sharing scheme.
A cost sharing scheme $\xi$ is \emph{semi-cross monotonic} if for every $S\subseteq A$ an player $i\in S$ all $j\in S\setminus\{i\}$ : either
$\xi(j,S\setminus\{i\})\leq\xi(j,S)$ or $\xi(j,S\setminus\{i\})\geq\xi(j,S)$. Notice that every cross monotonic cost sharing scheme is also semi-cross monotonic, since the second or condition is always true, however the converse does not hold.

As we later show in Proposition~\ref{prop:neg-neu} (a) Semi-cross-mo\-no\-to\-ni\-ci\-ty can be almost directly derived from the condition of \lcm we define in this work and more specifically from part (a) of \lcmd 

\section{Our Characterization}
\subsection{\lcm }
\lcm  considers each time a restriction of the mechanism that can only output as the serviced set, subsets of $U$ that contain all players in $L$. To be more formal consider all possible subsets of the players $L,U$ such that $L\subseteq U\subseteq \mathcal{A}$. Fixing a pair $L\subseteq U$ \lcm considers only sets of players $S$ with $L\subseteq S\subseteq U$.

Let $\m{i}$ be the minimum payment of player $i$, for getting serviced when the output of the mechanism is between $L$ and $U$, i.e. $\m{i}:=\min_{\{L\subseteq S\subseteq U,i\in S\}}\xi(i,S)$.

\begin{defn}[Fence Monotonicity]
We will say that a cost-sharing scheme satisfies \emph{\lcm } if it satisfies the following three conditions:
\begin{itemize}
\item[(a)] There exists at least one  set  $S$ with $L\subseteq S\subseteq U$, such that for all $i\in S$ we have $\xi(i,S)=\m{i}$.

\item[(b)] For each player $i\in U\setminus L$ there exists at least one set $S_i$, with $L\subseteq S_i\subseteq U$, such that for all $j\in S_i\setminus L$, we have $\xi(j,S_i)=\m{j}$.

(Note that $i \in S_i\setminus L$ and thus $\xi(i,S_i)=\m{i}$. Also note that we might have $S_i\neq S_j$ for $i\neq j$.)

\item[(c)] If there exists a set $C\subset U$, such that for some player $i$ we have $i\in C$ and $\xi(i,C)<\m{i}$ (obviously $L\nsubseteq C$), then there exists at least one set $T\neq \emptyset$, with $T\subseteq L\setminus C$ such that for all $j\in T$, $\xi(j,C\cup T)=\m{j}$.
\end{itemize}
\end{defn}

The first condition says that if we necessarily have to service the players in $L$, there exists a superset $S$ of $L$, such that if all players in $S$ such that the serviced players achieve their lowest possible (non-zero) payment $\m{i}$. As we show in Proposition~\ref{prop:neg-neu}(a), this condition generalizes the condition of semi-cross-monotonicity, which was identified as necessary for group-strategyproofness in~\cite{imm08} and in this sense our work completes the partial characterization obtained in~\cite{imm08}.

The second condition says then for each player $i\in U\setminus L$ there exists an outcome $S_i$, such that $i\in S_i$, and such that all players in $S_i\setminus L$ are served with their lowest possible payment. Loosely speaking this gives a way to enlarge $L$ by adding to it more players in a way that is optimal for the players that we add. Note that the cost of the players already in $L$ might however increase, which in fact relaxes cross-monotonicity, in the sense that if we further restrict the second condition to hold for every $j \in S_i$, then the underlying cost sharing schemes are cross-monotonic.

The third condition compares the minimum possible non-zero payment of each player in this restriction of the mechanism, with his minimum possible non-zero payment when the outcome can be any subset of $U$ (i.e. the the output should not necessarily contain the players in $L$). If the first payment is bigger then this means that some of the players in $L$ are responsible for this higher payment and ``harm'' the player.   Very loosely speaking condition (c) says that
at least one non-empty subset of $L\setminus C$, does not get ``harmed"
by a coalition that restricts the outcome to be a subset of $L\cup C$ (we remove all the players in $U\setminus (L\cup C)$ from $U$) and contain every player in $C$ (we add the players in $C\setminus L$ to $L$). The intuition of the third condition will become more clear when we show some important allocation properties of GSP mechanisms.

\begin{theorem}\label{theo:gsp=>bsm}
A cost sharing scheme gives rise to a group-strategyproof mechanism if and only if it satisfies \lcmd
\end{theorem}

\subsubsection{Examples of Mechanisms that violate just one part of \lcm and are not GSP}

We will give three representative examples to illustrate, why a cost sharing scheme, which does not satisfy \lcmc cannot give rise to a group-stra\-te\-gy\-proof mechanism. We chose our examples in a way that  only one condition of \lcm is violated and only at a specific pair $L,U$. (in fact in condition (c), the violation is present at two pairs, however it can be shown that this is unavoidable.)

\begin{example}[a]\label{ex:a}

Let $\A=\{1,2,3,4\}$. We construct a cost sharing scheme, such that condition (a) of \lcm is not satisfied at $L=\{1,2\}$ and $U=\{1,2,3,4\}$, as follows.

$$\begin{array}{|c||c|c|c|c|}
\hline
\xi & 1 & 2& 3 & 4 \\\hline\hline
\mathbf{\{1,2,3,4\}} & \mathbf{30}& \mathbf{30} & \mathbf{30} & \mathbf{30}\\\hline
\mathbf{\{1,2,3\}}&\mathbf{20}&\mathbf{30}& \mathbf{30} &\mathbf{-}\\\hline
 \mathbf{\{1,2,4\}}&\mathbf{30}&\mathbf{20}&\mathbf{-}&\mathbf{30}\\\hline
\{1,3,4\}& 30 & -& 20 &30 \\\hline
\{2,3,4\}&- & 30&30& 20\\\hline
\end{array}
\hspace{5pt}
\begin{array}{|c||c|c|c|c|}
\hline
\xi & 1 & 2& 3 & 4 \\\hline\hline
 \mathbf{\{1,2\}}& \mathbf{30} & \mathbf{30} &\mathbf{-}&\mathbf{-}\\\hline
\{1,3\}& 20 & - &30&-\\\hline
\{1,4\} & 30& 30 & - & -\\\hline
\{2,3\}&- & 30&30&-\\\hline
\{2,4\}&- & 20&-&30\\\hline
\end{array}
\hspace{5pt}
\begin{array}{|c||c|c|c|c|}
\hline
\xi & 1 & 2& 3 & 4 \\\hline\hline
\{3,4\}& - & -&30&30\\\hline
\{1\}& 30 & -& -&-\\\hline
\{2\}& - & 30 &-&-\\\hline
\{3\} & -& - & 30 & -\\\hline
\{4\}& - & -& - &30 \\\hline
\end{array}
\hspace{5pt}
$$

Consider the bid vector $b: = (b^*1, b^*_2, 30, 30).$ Notice that players 3 and 4 are indifferent to being serviced or not, as the single value they may be charged as payment equals their bid. Moreover, notice that either player 1 or player 2 (or both) must pay 30 strictly over their minimum payment 20 under
this restriction. Without loss of generality assume that $\xi(i,O(b)) = 30$. Consider the bid vector $b' := (b^*_1 , b^*_2 , b^*_3, -1)$. By VP and CS it
holds that $O(b') = \{1, 2, 3\}$ and thus $\xi(1,O(b')) < \xi(1,O(b))$.
Notice that the utilities of players $3$ and $4$ remain zero, and
consequently $\{1, 3, 4\}$ form a successful coalition.
In a similar manner we prove the  existence of successful coalition when $\xi(2,O(b)) = 30$.

\end{example}

\begin{example}[b]\label{ex:b}

Let $\A=\{1,2,3,4\}$. We construct a cost sharing scheme, such that condition (b) of \lcm is not satisfied at $L=\{1,2\}$ and $U=\{1,2,3,4\}$ for player $3$, as follows.

$$\begin{array}{|c||c|c|c|c|}
\hline
\xi & 1 & 2& 3 & 4 \\\hline\hline
\mathbf{\{1,2,3,4\}} &  \mathbf{30}&  \mathbf{30} &  \mathbf{30} &  \mathbf{30}\\\hline
 \mathbf{\{1,2,3\}}& \mathbf{30}&  \mathbf{30}&  \mathbf{40} &-\\\hline
 \mathbf{\{1,2,4\}}& \mathbf{30}& \mathbf{30}&-& \mathbf{20}\\\hline
\{1,3,4\}& 30 & -& 30 &30 \\\hline
\{2,3,4\}&- & 30&30& 30\\\hline
\end{array}
\hspace{5pt}
\begin{array}{|c||c|c|c|c|}
\hline
\xi & 1 & 2& 3 & 4 \\\hline\hline
 \mathbf{\{1,2\}}&  \mathbf{30} &  \mathbf{30} &-&-\\\hline
\{1,3\}& 30 & - &30&-\\\hline
\{1,4\} & 30& - & - & 30\\\hline
\{2,3\}&- & 30&30&-\\\hline
\{2,4\}&- & 30&-&30\\\hline
\end{array}
\hspace{5pt}
\begin{array}{|c||c|c|c|c|}
\hline
\xi & 1 & 2& 3 & 4 \\\hline\hline
\{3,4\}& - & -&30&30\\\hline
\{1\}& 30 & -& -&-\\\hline
\{2\}& - & 30 &-&-\\\hline
\{3\} & -& - & 30 & -\\\hline
\{4\}& - & -& - &30 \\\hline
\end{array}
$$

Consider the bid vector $b^3 := (b^*_1, b^*_2, 35, b^*_4 )$. Strategy\-proof\-ness implies that $3 \in O(b^3)$, since otherwise if she is not
serviced (zero utility), she can misreport $b_3^*$
changing the outcome to $\{1, 2, 3, 4\}$ and increasing her util\-ity to $35 -30>0$.

Next, consider the bid vector $b^4 :=
(b^*_1 ; b^*_2, 30, 25)$.
Assume that $4 \in O(b^4)$. Moreover, notice that by VP
it is impossible that $3 \in O(b^4)$. Thus, $\{3, 4\}$ can form a
successful coalition bidding $b' = (b^*_1, b^*_2,-1,b^*_4 )$, changing
the outcome to $\{1, 2, 4\}$ increasing the utility of player $4$ to
$25-20>0$, while keeping the utility of player $3$ at zero.

Finally, consider the bid vector $b^{3,4} := (b^*_1, b^*_2 , 35, 25)$.
Notice that the $b^{3,4}$ differs with $b^3$ and $b^4$ in the coordinates
that correspond to players $4$ and $3$ respectively. Like in the
case of $b^3$ the only possible outcomes by VP and CS at $b^{3,4}$
are $\{1, 2, 4\}$ and $\{1, 2\}$.

Assume that player 4 is serviced at $b^{3,4}$, which implies
that $\xi(4,O(b^{3,4})) < \xi(4,O(b^3))$. This contradicts strategy\-proofness, since when the true values
are $b^3$, player $4$ can bid according to $b^{3,4}$ in order to decrease
her payment and still being serviced.

Now, assume that player $4$ is not serviced at $b^{3,4}$. Then
$\{3, 4\}$ can form a successful coalition when true values are
$b^{3,4}$ bidding $b^4$, increasing the utility of player $4$ to $25-20>0$, while keeping the utility of player 3 at zero.

\end{example}

\begin{example}[c]

Let $\A=\{1,2,3,4\}$. This time we construct a cost sharing scheme, such that part (c) is not satisfied for $L=\{1,2\}$ (or $\{1,2,3\}$)
and $U=\{1,2,3,4\}$ and specifically for $C=\{3,4\}$ and $j=3$.

$$\begin{array}{|c||c|c|c|c|}
\hline
\xi & 1 & 2& 3 & 4 \\\hline\hline
\mathbf{\{1,2,3,4\}} & \mathbf{30}& \mathbf{30} & \mathbf{30} & \mathbf{30}\\\hline
\mathbf{\{1,2,3\}}&\mathbf{20}&\mathbf{20}& \mathbf{30} &\mathbf{-}\\\hline
 \mathbf{\{1,2,4\}}&\mathbf{30}&\mathbf{30}&\mathbf{-}&\mathbf{30}\\\hline
\{1,3,4\}& 30 & -& 30 &30 \\\hline
\{2,3,4\}&- & 30&30& 30\\\hline
\end{array}
\hspace{5pt}
\begin{array}{|c||c|c|c|c|}
\hline
\xi & 1 & 2& 3 & 4 \\\hline\hline
 \mathbf{\{1,2\}}& \mathbf{20} & \mathbf{20} &\mathbf{-}&\mathbf{-}\\\hline
\{1,3\}& 20 & - &20&-\\\hline
\{1,4\} & 30& 30 & - & -\\\hline
\{2,3\}&- & 20&20&-\\\hline
\{2,4\}&- & 30&-&30\\\hline
\end{array}
\hspace{5pt}
\begin{array}{|c||c|c|c|c|}
\hline
\xi & 1 & 2& 3 & 4 \\\hline\hline
\{3,4\}& - & -&\mathit{20}&30\\\hline
\{1\}& 30 & -& -&-\\\hline
\{2\}& - & 30 &-&-\\\hline
\{3\} & -& - & 30 & -\\\hline
\{4\}& - & -& - &30 \\\hline
\end{array}
\hspace{5pt}
$$
Suppose that the values are $b := (25, 25, b^*_3 ,30)$. The
only feasible by VP outcomes are  $\{1, 2, 3\}, \{1, 3\}$, $\{2,3\}$, $\{3, 4\}$ and $\{3\}$. Notice that player $4$ has zero utility regardless of
the outcome.

Assume that $O(b) \neq \{1, 2, 3\}$ and w.l.o.g that $1 \notin  O(b)$.
Then $\{1, 2, 4\}$ can form a successful coalition bidding $b' :=
(b^*_1, b^*_2, b^*_3, -1)$ increasing the utility of player $1$ to $25>0$ without
 decreasing the utility of player 2 (either $O(b) = \{2, 3\}$ and $\xi(2,O(b')) = \xi(2,O(b))$ and her utility remains the same utility or
$2\notin O(b)$ and her utility increases to $25>20$, like in the case of player $1$) and keeping the utility of player 4 at zero.

Finally, assume that $O(b) = \{1, 2, 3\}$. Then $\{3, 4\}$ can
form a successful coalition bidding $b'':= (25, 25, b^*_3, b^*_4 )$.
Obviously $\{3, 4\} \subseteq  O(b'')$. Notice that  VP excludes each of the following outcomes: $\{1, 3, 4\},$ $\{2, 3, 4\}$ and $\{1, 2, 3, 4\}$. As a
result, $O(b'') = \{3, 4\}$ implying that the utility of player 3
increases, as $\xi(3,O(b)) > \xi(3,O(b''))$, while player 4 keeps
her utility at zero.
\end{example}

If  you are given a cross-monotone cost sharing scheme it is rather straightforward how to construct a Moulin mechanism.
However if you are given a cost sharing scheme that satisfies \lcm it is not straightforward how to construct a GSP mechanism. \lcm should be coupled with an allocation rule that satisfies a simple property, which we call stability and a valid tie-breaking rule in order for the mechanisms to be GSP.

\subsection{Fencing Mechanisms}
Given a cost sharing scheme $\xi$ that satisfies \lcm we construct a mechanism, that uses $\xi$ as payment function and satisfies group-stra\-te\-gy\-proofness.

The mechanism takes as input the bids of the players and determines a pair of sets $L,U$ where $L\subseteq U$, where $L$ is the set of players that are going to be serviced by the mechanism with strictly positive utility. On the other hand $U\setminus L$ is the set of players that are indifferent between getting serviced or not, because their bid equals their payment, and we use a tie-breaking policy to determine which of these players will get serviced. The existence of a tie-breaking policy that does not violate group-stra\-te\-gy\-proofness is guaranteed by part (a) of \lcmd The intuition behind the tie-breaking rule is that it is optimal for the players in $L$, in the sense that from all subsets of $U$ we choose to serve the one where the players in $L$ achieve their minimum payments provided that they all get serviced.

The mechanisms we design can be put in the following general framework: Given a bid vector as input, we search for a certain pair of sets $L,U$, where $L\subseteq U\subseteq A$ that meet the criteria of stability we define below and then we choose one of the allocations that service the players in $L$ according to a valid tie-breaking rule. If the search is exhaustive the resulting algorithm is exponential time and we do not know any polynomial-time algorithm. However if we restrict our attention to payments that satisfy certain conditions like for example cross-monotonicity we can come up with a polynomial-time algorithm for finding a stable pair.
\begin{defn}[Stability]
A pair $L,U$ is \emph{stable} at $\bv$ if the following conditions are true:
\begin{enumerate}
\item For all $i\in L$, $b_i>\m{i}$,
\item for all $i\in U\setminus L$, $b_i=\m{i}$ and
\item for all $R\subseteq A\setminus U$, there is some $i\in R$, such that
$b_i<\mlu{i}{L}{U\cup R}$.
\end{enumerate}
\end{defn}

\begin{remark}
Assume that $\xi$ is Cross-monotonic and let $S$ be the output of Moulin mechanism for some bid vector $\bv$. Then, the pair $L,U$, where $L=\{i\in S\mid b_i>\xi(i,S)\}$ and $U=S$,
is the unique stable pair at $\bv$.
\end{remark}

After identifying a stable pair these mechanisms output a set $S$, where $L\subseteq S\subseteq U$ given by a tie-breaking function.

\begin{defn}
The mapping $\sigma:2^\A\times 2^\A \times \mathbb{R}^n \rightarrow 2^{\A}$ is a \emph{valid} tie-breaking rule for $\xi$,
if for all $b$ and $L\subseteq U\subseteq A$, such that $L,U$ is stable at $b$, for the  set $S=\sigma(L,U,b)$ it holds that $L\subseteq S\subseteq U$ and for all $i\in S$, $\xi(i,S)=\m{i}$. Part (a) of \lcm guarantees that there exists at least one output $S$ with this property and each different such output gives a different tie-breaking rule.
\end{defn}

\begin{defn}
We will say that a mechanism is a \emph{Fencing Mechanism} if at input $\bv$ it finds a stable pair $L,U$ at $b$, outputs $\sigma(L,U,b)$, where $\sigma$ is a valid tie-breaking rule and charges each player $i$,  the value $\xi(i,\sigma(L,U,b))$.
\end{defn}

It is easy to verify that every Fencing Mechanism satisfies VP, because the players that get serviced belong to the set $U$ and as the mechanism satisfies stability these players have non-negative utility, and CS, because if a player bids higher than any of his payments, then again by stability he belongs to the set $L$ and gets serviced.

\begin{remark}
Assume that for two distinct bid vectors $b$ and $b'$, the pair $L,U$ is stable. Then for a Valid tie-breaking rule it may hold that   $\sigma(L,U,b)\neq\sigma(L,U,b')$, which implies that the outcome of a Fencing mechanism may change, though the utilities of the players remain unchanged. If we are not interested in considering the whole class of GSP allocations, that arise from cost sharing scheme that satisfies \lcmc we can assume that the tie-breaking rule depends only on the sets $L$ and $U$.
\end{remark}

\begin{remark}
Moulin mechanisms are GSP and conseque\-nt\-ly they can be viewed as special case of the general framework of Fencing Mechanisms. In Moulin mechanisms $\xi$ is cross-monotonic and consequently the bigger the set of serviced players, the lower is the cost for each one of them. Thus it holds that for all $L\subseteq U\subseteq \A$, that for all $i\in U$, $\xi(i,U)=\m{i}$ and the mapping $\sigma(L,U,b)=U$ for all $L\subseteq U\subseteq \A$ and all $b$ such that $L,U$ is stable, is a valid tie-breaking rule. This simplifies the algorithm substantially, as we just have to find a set $U$ of players that can be serviced with utility greater or equal to zero. Moreover as the payment of each player increases when the serviced players becomes smaller we need the set $U$ with maximal cardinality.
\end{remark}

\begin{theorem}\label{theo:gsp<=>bs}
A mechanism is group-strategyproof if and only if it is a Fencing Mechanism.
\end{theorem}

\section{Every GSP Mechanism is a Fencing Mechanism} \label{nec}


\subsection{Necessity of \lcm}
In this section we prove that the payment function of every GSP mechanism satisfies \lcmd  A natural approach would assume that one condition of \lcm  is violated at some $L,U$, and prove  that  it is impossible that the corresponding cost sharing scheme gives rise to a GSP mechanism. It turns out that our lack of knowledge of the payment function renders this approach  unlikely to be fruitful.

Therefore, we will follow an alternative method. We  select an arbitrary GSP mechanism and consider some $U\subseteq \A$. We show that for every $L\subseteq U$ the cost sharing scheme  satisfies each one of the three conditions of \lcm  using induction on $|U \setminus L|$.  When proving the induction step we also reveal several important allocation properties for GSP mechanisms.

\textbf{Base:} For $|U \setminus L|=0$, i.e $L=U$ every part of \lcm  is trivially  satisfied as follows.

Condition (a): It holds that $\mlu{i}{U}{U}=\xi(i,U)$, since the minimum in the definition of $\xi$ is taken over the single possible outcome $U$.

Condition (b): This condition holds trivially as $U\setminus L=\emptyset$.

Condition (c):
Regardless of whether the if condition of this part is true, it holds that for all $C\subset U$ we can set
$T=U \setminus C$ since for all $j\in T$, $\xi(j, C\cup T)=\xi(j, U)=\mlu{j}{L}{U}$. \

\textbf{ Induction Step:}  Proving the induction step requires some definitions that allows the effective use of the induction hypothesis in order to identify a successful coalition if some part is violated.   We first define the notion of a harm relation and we prove that it is a strict partial order.

\subsubsection*{Harm relation}

\begin{lemma}\label{lem:xicompare}
 If  $U\subseteq U_1$ and $L_1\subseteq L$, then for all $i \in U$,  $\m{i}\geq\mlu{i}{L_1}{U_1}$.
\end{lemma}

\begin{defn}[Harm]
 We say  that $i$ harms $j$, where $i,j\in U$  if and only if  $\mlu{j}{L}{U}<\mlu{j}{L\cup\{i\}}{U}$
\end{defn}

Consequently, for all distinct $i,j\in U$, $i$ either \emph{harms} $j$ or
otherwise it holds that $\m{j}=\mlu{j}{L\cup \{i\}}{U}$ (from Lemma \ref{lem:xicompare}). Trivially every $i \in L$ \emph{does not harm} any other player $j \in U$.

\begin{claim}\label{cla:harm}
The harm relation  satisfies anti-symmetry and transitivity and consequently it is strict partial order.
\end{claim}

\begin{corollary}\label{cor:harmdag}
The induced sub-graph $G[U \setminus L]$ is a directed acyclic graph.
\end{corollary}

\subsubsection*{Condition (a) of \lcm }

The core idea that we will use for the proof of conditions (a) and (b) is to construct bid vectors, where VP and CS restrict the possible outcomes to be subsets of $U$ that contain every player in $L$, while for condition (c) we will restrict the possible outcomes to be subsets of $U$. Therefore, we assume that every player in $L$ has bidden a very high value and and every player in $\A\setminus U$ has bidden a negative value.

 For the proof of condition (a), we also want the players in $U\setminus L$ to be indifferent between being serviced and getting excluded from the outcome, i.e. they have zero utility.  Thus, we assume that every player in $U \setminus L$ has bidden exactly her minimum payment $\xi^*$ at $L,U$.

We first use the induction hypothesis and the properties of the harm relation to prove that the following Lemma, which is a condition somewhat milder than condition (a) of \lcmd
\begin{lemma}\label{lem:a-milder}
For every $j \in L$, there is some set $S_j$, where $L\subseteq S_j \subseteq  U$ such that for all $ i \in S_j \setminus L$, $\xi(i , S_j)=\m{i}$ and $\xi(j, S_j)= \m{j}$.
\end{lemma}
Then we use the preceding Lemma to show that if the cost-sharing scheme does not satisfy condition (a) of \lcm then there exists a successful coalition.
\begin{lemma}\label{lem:a-alloc}
At the bid vector $b$, where for all $i \in L$, $b_i=b^*_i$, for all $i \in U \setminus L$, $b^*_i= \m{i}$ and for all $i \notin U$, $b_i =-1$, it holds that $L\subseteq O(b)\subseteq U$ and that for all $i \in O(b)$, $\xi(j, O(b))=\m{j}$.
Setting $S=O(b)$, condition (a) of \lcm  is satisfied at $L,U$.
\end{lemma}

\subsubsection*{Condition (b) of \lcm }
We consider now the players in $U \setminus L$. Using the induced sub-graph $G[U\setminus L]$  of the harm relation, we can discriminate them by whether a player is \emph{sink} of this graph or not.  First we consider the sinks, as the  satisfaction of the second condition of  \lcm for a sink  is an immediate consequence of the induction hypothesis.

\begin{claim}\label{cla:sinkandb}
For every sink $k$ of  $G[U\setminus L]$ condition (b) of \lcm  is satisfied at $L,U$.
\end{claim}

We continue with the rest players in $U \setminus L$.

\begin{claim}\label{claim:harmandU-L}
For every $j \in U \setminus L$ one of the following holds:  either $j$ is a sink of the sub-graph $G[U \setminus L]$,  or there is a sink $k$ such that $j$ harms $k$.
\end{claim}

Now consider an player $j$ in $U\setminus L$, that is \emph{not a  sink} of $G[U\setminus L]$ and let $k$ be one of its \emph{sinks} such that $j$ \emph{harms} $k$. In order to prove that the second condition is  satisfied for $j$, we will involve group-stra\-te\-gy\-proofness at certain bid vectors (trying to generalize Example \ref{ex:b} where $j$ takes the role of player $3$ and $k$ the role of player $4$). Prior to defining these inputs, we prove another allocation property.
 
\begin{lemma}\label{lem:b-alloc}
Consider some $L'\subseteq U'\subseteq \A$. Assume that the set $S_j$, as in the definition of condition (b) of \lcm   exists for some  $j\in U'\setminus L'$.  At any bid vector $b^j$, such that for all $i \in L'$, $b^j_i=b^*_i$, $b^j_j> \mlu{j}{L'}{U'}$ for all $i \in U' \setminus (L'\cup\{j\})$, $b^j_i= \mlu{i}{L'}{U'}$ and for all $i \notin U$, $b^j_i =-1$, it holds that $j \in O(b^j)$ and  $\xi(j, O(b^j))=\mlu{j}{L'}{U'}$.
\end{lemma}

Let $\epsilon$ be a very small positive number, smaller than any positive payment difference. We construct two bid vectors, at which we can characterize the allocation of the mechanism. First, consider the bid vector $b^k$, where for all $i \in L$, $b^k_i=b^*_i$, $b^k_k=\m{k}+\epsilon$, for all $i \in U \setminus (L\cup\{k\})$, $b^k_i=\m{k}$ and for all $i\notin U$,
$b^k_i=-1$.

\begin{claim} \label{claim:bk-alloc}
At the bid vector $b^k$ the following hold

(a) Player $k$ is serviced  and charged $\m{k}$.

(b) Player $j$ is not serviced.
\end{claim}

Second, consider the bid vector $b^j$, where for all $i \in L\cup\{k\}$, $b^k_i=b^*_i$, $b^j_j=\m{j}+\epsilon$, for all $i \in U \setminus (L\cup\{i,j\})$, $b^j_i=\m{i}$ and for all $i\notin U$,
$b^j_i=-1$.

\begin{claim}\label{claim:bj-alloc}
At the bid vector $b^j$ the following hold

(a) For all $i \in U\setminus (L\cup\{j,k\})$, $b_i=\mlu{i}{L\cup\{k\}}{U}$ and $b_j=\mlu{j}{L\cup\{k\}}{U}+\epsilon$.

(b) Player $j$ is serviced  and charged $\m{j}$.

(c) Player $k$ is  serviced and charged  more than $\m{k}$.

\end{claim}

Finally, we construct the  intermediate bid vector $b^{j,k}$, where  player $j$ bids according to $b^j$ ($\m{j}+\epsilon$) player $k$ bids according to $b^k$ ($\m{k}+\epsilon$), and every other player bids the same value as in both bid vectors.

\begin{claim} \label{claim:bjk-alloc}
At the bid vector $b^{j,k}$ the following hold

(a) Player $k$ is not serviced at $b^{j,k}$.

(b) Player $j$ is serviced at $b^{j,k}$.

(c) $L\subseteq O(b^{j,k})\subseteq U$  and every player $i \in O(b^{j,k})\setminus L$, is charged $\m{i}$.
\end{claim}

As a result setting $S_j=O(b^k)$, we conclude that the second condition is satisfied for any $j \in U\setminus L$ that is not a sink of $G[U\setminus L]$ as well.

\subsubsection*{Condition (c) of \lcm }

To show that the cost-sharing scheme satisfies the third property of \lcmc at $L,U$, we need the induction hypothesis only for showing (as we have already done) that condition (a) of \lcm  is  satisfied  at this pair and specifically the allocation properties of Lemma~\ref{lem:a-alloc}.

  The main idea is to define two families of bid vectors, that both contain the previous special case  and moreover  the first family of inputs is a subset of second.

\begin{center}
\begin{tabular}{|c|c|c|c|}
\hline
& $L$ & $U\setminus L$ & $\A\setminus U$ \\\hline\hline
Special Case & $b^*_i$ & $\m{i}$ & $-1$ \\\hline
First Family & $>\m{i}$ & $\m{i}$ & $-1$\\\hline
Second Family & $>\m{i}$ & $\in\mathbb{R} $& $-1$\\\hline
\end{tabular}
\end{center}

The use of induction was one of the basic techniques in~\cite{imm08}, however here we need to use induction in a more powerful way. In~\cite{imm08} the authors first fix an ordering of the players and then apply induction, while here we start from a bid vector that satisfies a certain property (induction base) and use induction on the number of coordinates at which the new bid vector differs from the bid vector used in the induction base. This allows us to use the induction hypothesis more effectively, as we can alter the coordinates of our choice first, instead of selecting an ordering and then formulating the induction statement. While  the proof of the allocation properties about the first family does not  require the advantage of this technique, it is the essence of our proof of the corresponding allocation properties about the second family.

 Notice that the special input differs from the bid vectors of the first family  only in the bid coordinates that correspond to players in $L$.  We apply our technique and extend the implications of Lemma~\ref{lem:a-alloc} for every bid vector in the first family.

\begin{lemma}\label{lem:char1}
 For every bid vector $b$, where for all $i \in L$, $b_i>\xi(i, S)$, for all $i \in U \setminus L$,
$b_i =\m{i}$ and for all $i \notin U$, $b_i=-1$, it holds that   $L\subseteq O(b)\subseteq U$ and for all $i \in O(b)$, $\xi(i, O(b))=\m{i}$.
\end{lemma}

Next we provide a weaker allocation property about the inputs of the second family. We arbitrarily select a vector that belongs to the first family and then apply our technique for proving this property for every input that is ``reachable" by our initial vector by altering the coordinates that correspond to the players in $U\setminus L$. The arbitrary selection of the initial vector ensures that our statement holds for every input that belongs to the second family.

\begin{lemma}\label{lem:char2}
 For every bid vector $b$, where for all $i \in L$, $b_i>\xi(i, S)$, for all $i \in U \setminus L$, $b_i \in \mathbb{R}$ and for all $i \notin U$, $b_i=-1$, it holds that:  for all $i \in O(b)$, $\xi(i, O(b))\geq\m{i}$.
\end{lemma}

This Lemma may shed some light on the understanding of the GSP mechanisms and  can be interpreted as follows: Assume that the players in $\A \setminus U$ are uninterested to participate. Now if all the bids of players in a set  $L\subseteq U$ have surpassed their respective minimum payments at $L,U$, then a GSP mechanism never excludes a group of players  in $L$ from the outcome in order to charge a serviced player less than the restriction of their presence. Loosely speaking the players in $L$ ``fence" any outcome $C\subset U$ such that there is some $j \in C$ such that $\xi(j,C)<\m{j}$. Since this ``fencing" phenomenon must be true for arbitrary bids of the players in  $U\setminus L$, it follows that every GSP cost-sharing scheme must satisfy condition (c) of \lcmc as shown below.

\begin{claim}\label{claim:bc-alloc}
We construct the bid vector $b^C$ as follows: For all $i \in C$, $b^C_i=b^*_i$, for all $i \in L\setminus C$, $b^C_i=\m{i}+\epsilon$ and for all $i  \notin C\cup L$, $b^C_i=-1$.

For the bid vector $b^C$ it holds that:

(a) For all $i\in O(b^C)$ it holds that $\xi(i,O(b^C))\geq\m{i}$.

(b) $C\subset O(b^C) \subseteq L\cup C$.

(c) For all $i \in O(b^C) \setminus C$, it holds that
 $\xi(i,O(b^C))=\m{i}$.
\end{claim}

Setting $T=O(b^C)\setminus C$ we complete the proof of the third condition.
\subsection{Necessity of Stability and Valid tie-breaking}
In the last part, we show that the allocation of every GSP mechanism satisfies Stability and uses a Valid tie-breaking rule.  We already know that the payment function satisfies \lcm and thus we can prove the following generalization of Lemma \ref{lem:char1}.

\begin{lemma}\label{lem:char3}
Let $L\subseteq U\subseteq \A$. For every bid vector $b$, such that for all $i \in L$, $b_i>\m{i}$, for all $i \in U\setminus L$, $b_i=\m{i}$ and for all $R\subseteq \A\setminus U$, there is some $i \in R$ such that $b_i<\mlu{i}{L}{U\cup R}$, it holds that $L\subseteq O(b)\subseteq U$ and for all $i \in O(b)$, $\xi(i,O(b))=\m{i}$.
\end{lemma}

Lemma \ref{lem:exists}  implies  that there is a stable pair at every input. Thus, given a bid vector, we may apply Lemma \ref{lem:char3} with $L,U$  being the corresponding stable pair at this input, and get that $L\subseteq O(b)\subseteq U$ (Stability) and for all $i \in O(b)$, $\xi(i,O(b))=\m{i}$ (Validity of the tie-breaking rule).

\section{The classes of GSP and Fencing Mechanisms coincide}

In this section, we complete our characterization by proving that  Fencing Mechanisms are GSP.

\subsection{Properties of Stable pairs}
\begin{lemma}\label{lem:maxu}
For every bid vector $b$ and set $L$ with $L\subseteq \A$, there exists a unique maximal set $U$, $U \supseteq L$, such that for all $i \in U \setminus L$ we have $b_i \geq \m{i}$ and any other set with the same property is a subset of $U$.
Moreover, the pair $L,U$  satisfies condition 3. of Stability
\end{lemma}

Consider all possible outcomes of the mechanism, that contain every player in $L$. Now we remove all the sets $S$, where at least one player $i \in S\setminus L$, has bidden strictly less than her payment in $S$. The set $U$ is the union of all the sets that remain after this filtering. Notice, that this is always  true   only if the underlying cost sharing scheme satisfies condition (b) of \lcmd

Furthermore, notice that we haven't yet shown that that there is  a stable pair at every input. The next two properties will be used together in our proofs as a criterion whether a pair $L,U$ is stable at a given bid vector $b$. 

\begin{lemma}\label{lem:uncov2}
Suppose that $L,U$ is a stable pair at the bid vector $\bv$ and that $S$ is set with the property that for all $i\in S$ we have that $b_i-\xi(i,S)\geq 0$. (If the mechanism would output $S$ then all players are would have been served with non-negative utility.)  If

(a) $S\not\subseteq U$, or

(b) if $S\subseteq U$ and for some $i\in S$ we have $\m{i}>\xi(i,S)$,
\\
there exists some non-empty set $T\subseteq L\setminus S$, such that for all $j\in T$ we have $\xi(j,S\cup T)<b_j$.
\end{lemma}
This lemma is an immediate consequence of part (c) of \lcm  and the definition of stability.

\begin{lemma}\label{lem:cov3}
Suppose that $L\subseteq S\subseteq U$ and that there exists a non-empty $T\subseteq \A\setminus S$ such that for all $i\in T$ we have $b_i\geq \xi(i,S\cup T)$ and that for at least one player from $T$ the inequality is strict.
Then $L,U$ is not a stable set at the bid vector $\bv$.
\end{lemma}

\subsection{Uniqueness and group-stra\-te\-gy\-proofness}
\begin{lemma}\label{lem:unique}
If for some bid vector $\bv$ there exists a stable pair then it is unique.
\end{lemma}

A crucial point of our proof is that for the following step we assume that if there exists a stable pair the mechanism produces an outcome, otherwise it terminates without providing an answer. We will first show that the mechanism is GSP, wherever there exists a stable pair. Then we will use the fact that the mechanism satisfies group-stra\-te\-gy\-proofness for inputs that have a stable pair, to prove the existence of a stable pair for every input.

\begin{lemma}
For the inputs where there exists a stable pair, every Fencing Mechanism is group-strategyproof.
\end{lemma}

\begin{proof}
Let $\bv$ and $\bpv$ be two bid vectors, and let $L,U$ and $L',U'$ be their corresponding unique (from Lemma \ref{lem:unique}) stable pairs and $O(\bv)$ and $O(\bpv)$ the corresponding outputs. Assume towards a contradiction that some of the players can form a successful coalition when the true values are $\bv$ reporting $\bpv$.

We will first show that any player $i$ served in the new outcome, $i\in O(\bpv )$, has non-negative utility i.e. $b_i\geq\xi (i,O(\bpv ))$. Take some $i$ that is output in $O(\bpv)$. If $b_i=b'_i$ then it holds trivially since the mechanism satisfies VP at the outcome $O(\bpv)$. If $b_i\neq b'_i$ then $i$ changes his bid to be part of the coalition and consequently his utility after this coalition is non-negative, which gives $b_i-\xi(i,O(\bpv))\geq 0$.

The next step is to apply Lemma \ref{lem:uncov2} to show that there exists some non-empty $T\subseteq L\setminus O(\bpv)$ such that for all $i\in T$ we have $b_i>\xi(i,O(\bpv)\cup T)$. If $O(\bpv)\not\subseteq U$ then the premises of the Lemma hold trivially.

Suppose that $O(\bv)\subseteq U$. For the coalition to be successful the utility of at least one player $j$ increases strictly when the players bid $\bpv$ consequently $j\in O(\bpv)$ and $b_j-\xi(j,O(\bpv))> 0$. We will show that $\m{j}>\xi(j,O(\bpv))$. If $j$ is not served at $\bv$ and since $O(\bpv)\subseteq U$ from stability we get that $j\in U\setminus L$ and $b_j=\m{j}>\xi(j,O(\bpv))$. If $j$ is served at $\bv$ then her payment equals $\m{j}$ by the definition of the mechanism and in order that she profits strictly it must be $\xi(j,O(\bpv))<\m{j}$, so we can again apply Lemma \ref{lem:uncov2}.

Finally we will show that for all $i\in T$ we have $b_i=b'_i$.  After the manipulation the players in $T$ are not serviced, while as $T\subseteq L$ from stability we have that in the truthful scenario the players in $T$ are serviced with positive utility. Consequently the players in $T$ wouldn't have an incentive to be part of the coalition and change their bids.

Putting everything together we get that there exists a $T\subseteq \A \setminus O(\bpv)$ such that for all $i\in T$ we have $b'_i>\xi(i,O(\bpv) \cup T)$, which by Lemma \ref{lem:cov3} contradicts our initial assumption that $L',U'$ is stable at $b'$.
\end{proof}

\subsection{Existence of a stable pair for every input}
\begin{lemma}\label{lem:exists}
For every bid vector $\bv$ there exists a unique stable pair.
\end{lemma}
\begin{proof}
Let $b_i^*$ be some value that is big enough so that player $i$ always gets serviced, you can let $b_i^*>\max_{S\subseteq \A}\xi(i,S)$, since the allocation of the mechanism satisfies Stability. We will show that there exists a stable pair at any input $\bv$ by induction on the number $m$ of coordinates that are less than $b_i^*$, i.e. on the number $m=|\{i\mid b_i<b_i^*\}|$.

\textbf{Base:} For $m=0$, we only have to show that there exists a stable pair for the bid vector $(b_1^*,\ldots,b_n^*)$ and $\A,\A$ is a stable pair.

\textbf{Induction Step:} Suppose that if a bid vector has $m-1$ coordinates that are less than $b_i^*$, then it has a stable pair. We will show that if a bid vector $\bv$ has $m$ coordinates that are less than $b_i^*$ then it also has a stable pair. We will suppose towards a contradiction that there exists no stable pair at $\bv$.

We first need some definitions. Let $L^*:=\{i\mid b_i\geq b_i^*\}$ and $U^*$ the corresponding (from Lemma \ref{lem:maxu}) maximal set.  Notice
that for all $i \in L^*$, $b_i = b^*_i > \mlu{i}{L^*}{U^*}$ by the definition of $b^*_i$, which implies that the pair $L^*,U^*$ satisfies the first condition of stability. Moreover, from Lemma 6 we get that  the third condition  of stability is satisfied for $L^* , U^*$ as well. Thus, as $L^*,U^*$ cannot be stable at $\bv$ (from our assumption), the set $W=\{i\mid i\in U^*\setminus L^* \text{ and } b_i>\mlu{i}{L^*}{U^*}\}$ should be non-empty.

 For each $i\in W$ we define a corresponding pair $L_i,U_i$ as follows: The pair $L_i,U_i$ is the unique (by Lemma \ref{lem:unique}) stable pair of $(b_i^*,\bv_{-i})$, which exists by the induction hypothesis.
\begin{claim}\label{claim:inclusion*}
If $L_i,U_i$ is the stable pair of $(b_i^*,\bv_{-i})$, then

(a) $L^*\cup\{i\}\subseteq L_i$ and (b) $U_i\subseteq U^*$.

 (c) If $j\in L_i\setminus (L^*\cup \{i\})$ then $j\in W$ and $b_j>\xi^*(L^*,U^*)$ .
\end{claim}
\begin{claim}\label{claim:subset}
If there exists no stable pair at $\bv$, then for all $i\in W$ we have

(a) $b_i\leq \mlu{i}{L_i}{U_i}$.

(b)  $L^*\cup \{i\}\subset L_i$.

\end{claim}

Since $W\neq \emptyset$ there exists some $i\in W$ such that $L_i$ has minimum cardinality, i.e. $i=\arg\min_{i\in W}|L_i|$. By Claim \ref{claim:subset} (b) there exists some $j\in L_i\setminus (L^*\cup \{i\})$ and by Claim \ref{claim:inclusion*} (c) $j \in W$. We will show that $L_j\subset L_i$ contradicting the choice of $i$.

\begin{claim}\label{claim:notgsp}
If there exists no stable pair at $b$ and $j\in L_i\setminus (L^*\cup \{i\})$ then

(a) The pair $L_i,U_i$ is stable at the bid vector $(b_i^*,b_j^*,\bv_{-\{i,j\}})$.

(b) $\mlu{j}{L_j}{U_j}>\mlu{j}{L_i}{U_i}$ and $i\notin L_j$.

(c) If there exists some $k\in L_j\setminus L_i$, then the mechanism is not group-stra\-te\-gy\-proof at inputs where a stable pair exists.
\end{claim}

Now since $j\in W$ and from Claim \ref{claim:notgsp} $i \notin L_j$ and for every $k\in L_j$ we have that also $k\in L_i$, we get that $L_j\subset L_i$, which completes the proof.
\end{proof}

\section{A lower bound for the Budget Balance of any GSP mechanism}

In the last section we demonstrate the use of our character\-ization by showing that even in the case of three players 
there is a family of cost functions parameterized with a vari\-able $x$, where every GSP mechanism cannot achieve better
budget-balance than $\frac{1}{x}$.

First, we show some consequences of \lcm for small numbers of players, that will simplify the use of it.

\begin{proposition}\label{prop:neg-neu}
Let $\xi$ be a  cost sharing scheme, that satisfies \lcmd For all $S\subseteq\A$ and two distinct $i,j\in S$,  if $\xi(j,S \setminus \{i\})<\xi(j,S)$ then

(a) For all $k\in S\setminus \{i,j\}$, $\xi(k,S\setminus\{i\})\leq\xi(k,S)$

(b) $\xi(i,S\setminus \{j\})\leq\xi(i,S)$.

(c) $\xi(i,S\setminus \{j\})\geq\xi(i,S)$.

(part (a) is an alternative definition of  semi-cross monotonicity (\cite{imm08}))
\end{proposition}

Every part of the preceding Property, is implied by the corresponding condition of \lcmd Moreover, part (b) and (c) fully characterize the
case of two players. However, this Property is far from characterizing the case of the three players, as there are many
other constrains, that are not captured by its implications.

\begin{theorem}\label{theo:low}
Let $\A=\{1,2,3\}$. Consider the cost sharing function defined on $\A$ as follows: $C(\{1,2\})=C(\{1,3\})=1$, $C(\{1\})=C(\{2\})=C(\{3\})=x$,
$C(\{2,3\})=x^2+x$ and $C(\{1,2,3\})=x^3+x^2+x$, where $x\geq 1$.
There is no $\frac{1}{x}$-budget balanced cost sharing scheme, that satisfies \lcmd
\end{theorem}

It is important to understand that for the proof of this theorem,
we use every implication of Proposition 1 and thus every
condition of \lcmd

\section{Complexity and open questions}
Does there exist a polynomial-time algorithm for finding the allocation of a cost-sharing scheme that satisfies  \lcm  or maybe we can show that the problem of finding a stable pair is computationally hard? 

A natural question that arises in this context is whether, it is computationally more efficient to find the appropriate outcome than identifying the stable pair.
Suppose that you have an algorithm that computes the outcome of a GSP mechanism.
It is  rather straightforward how to compute the lower set $L$ of the stable pair, simply by checking which players have positive utility. What remains is to find the upper set $U$, which is exactly the maximal set, as defined in Lemma \ref{lem:maxu}. 

\begin{theorem}\label{theo:compl}
Suppose that we are given the outcome of a gr\-oup-str\-at\-egy\-pro\-of mechanism at $b$. Given that we have already computed $\m{i}$ for all $L\subseteq U\subseteq \A$ and all $i \in U$, there is a polynomial time algorithm for identifying its stable pair.
\end{theorem}

We believe that some other interesting directions for future research are the following: How can our characterization be applied for obtaining cost-sharing mechanism with better approximation and budget-balance guarantees or lower bounds for specific problems? And finally can our techniques be extended to obtain a characterization of weak-group-stra\-te\-gy\-proof cost-sharing schemes?

\section*{Aknowledgements}
We would like to thank Elias Koutsoupias for suggesting the problem, as well as for many very helpful insights and discussions. We would also like to thank Janina Brenner, Nicole Immorlica, Evangelos Markakis, Tim Roughgarden, and Florian Schoppmann for helpful discussions and some pointers in the bibliography.

\bibliographystyle{abbrv}
\bibliography{GSPbib}

\appendix

\section{Missing proofs of Section 5}

\subsubsection*{Proof of Lemma \ref{lem:xicompare}}

It is obvious that the minimum payment of player $i$ can only decrease as the set of outcomes, over which the minimum in the definition of $\xi^*$ is taken, becomes larger.\qed

\subsubsection*{Proof of Claim \ref{cla:harm}}

We show this using the induction hypothesis  at every  $L\cup\{i\},U$ for every $i \in U \setminus L$ and more specifically condition (c) of \lcmd

To show anti-symmetry we need to consider  only elements  of $ U\setminus L$, as  trivially it is impossible that some  $i \in L$ harms any other elements. Consider two distinct $i,j\in U \setminus L$ and assume that $i$ \emph{harms} $j$, that is $\m{j}<\mlu{j}{L\cup\{i\}}{U}$.  From definition of $\xi^*$ there is a set  $S_j$,  where $j \in S_j$,  $L\subseteq S_j\subseteq U$ and
$\xi(j,S_j)=\m{j}$. Notice that by our assumption it is impossible that $i \in S_j$, since $\mlu{j}{L\cup\{i\}}{U}>\xi(j,S_j)$. It follows that $S_j\subset U$  and by using condition (c) of \lcm  at $L\cup\{i\},U$, we get that
$\xi(i,S_j\cup\{i\})=\mlu{i}{L\cup\{i\}}{U}=\m{i}$ (the only non-empty subset of
$(L\cup\{i\}) \setminus S_j$ is $\{i\}$). Since $L\cup\{j\}\subseteq S_j\cup\{i\}\subseteq U$ we get that $\mlu{i}{L\cup\{j\}}{U}=\m{i}$.

We show transitivity of the harm relation in a similar manner. Consider three distinct players $i$, $j$ and $k$, where $i,j\in U \setminus L$ and $k \in U$ (may also belong to $L$). Assume now that $i$ \emph{harms} $j$, $j$ \emph{harms} $k$ while $i$ \emph{does not harm}  $k$. We will show that this contradicts our induction hypothesis. Since $\m{k}=\mlu{k}{L\cup\{i\}}{U}$, it follows that there is some  set $S_k$, where  $L\cup\{i\}\subseteq S_k\subseteq U$ such that
$\xi(k,S_k)=\m{k}$.  Using similar arguments like in the previous part we show that $xi(j,S_k\cup\{j\})=\m{j}$ and as $L\cup\{i\}\subseteq S_k\cup\{j\}\subseteq U$ we reach a contradiction from our assumption that $\m{j}<\mlu{j}{L\cup\{i\}}{U}$.
\qed

\subsubsection*{Proof of Lemma \ref{lem:a-milder}}
(a) We first have to prove the following Claim, which is an immediate consequence of the fact that the harm relation is a strict partial order.

\begin{claim}\label{claim:sinkandL}
For every $j \in L$ one of the following holds:  either every $i \in U \setminus L$ \emph{harms} $j$, or
 there a $k \in U \setminus L$ that \emph{does not harm} $j$ and also is a \emph{sink} of $G[U\setminus L]$.
\end{claim}
\begin{proof}[of Claim \ref{claim:sinkandL}]

Suppose that there is some $i \in U \setminus L$ that \emph{does not harm}  $j$.
If $i$ is a sink of $G[U\setminus L]$, setting $k=i$ completes our proof. Otherwise, there must be a path that goes through $i$ but does not stop there. Let $k$ be the sink of this path ($G[U\setminus L]$ is a directed acyclic graph).

Notice that transitivity implies that $i$ \emph{harms} $k$. Thus, it is impossible that $k$ \emph{ harms }$j$, since using transitivity again we would deduce that $i$ \emph{harms} $j$ contradicting our assumption. 
\end{proof}
Consider some $j \in L$. We split the proof in two cases as  Claim \ref{claim:sinkandL} indicates. 

\textbf{Case 1:} Suppose that every $ i \in U \setminus L$ \emph{harms} $j$. This implies that $\m{j}=\xi(j,L)$ and hence we can set $S_j=L$, since the requirements of  Lemma \ref{lem:a-alloc}  are trivially satisfied, since $S_j \setminus L=\emptyset$.

\textbf{Case 2:} Consider some sink $k$ that \emph{does not harm}  $i$. Using induction hypothesis at $L \cup \{k\},U$ and
in particular part (a) we get that there is a set $S$, where $L\cup\{k\}\subseteq S\subseteq U$ such that  for all $i \in S$, $\xi(i, S)=\mlu{i}{L\cup\{k\}}{U}$. Using the fact that $k$ is a sink and also \emph{does not harm}  $j$ we get that for all $i \in S' \setminus L$, $\mlu{i}{L\cup\{k\}}{U}=\m{i}$ and $\mlu{j}{L\cup\{k\}}{U}=\m{j}$. As a result, we can set $S_j=S$. \qed

\subsubsection*{Proof of Lemma \ref{lem:a-alloc}}
By CS and VP the output of the mechanism satisfies $L\subseteq O(b)\subseteq U$. First note that in any case a player $i\in U\setminus L$ has utility zero: either she is not serviced, or by VP, if she is serviced her payment cannot exceed her bid and cannot be less than her minimum payment $\m{i}$ so $\xi(i,O(b))= b_i=\mlu{i}{L}{U}$.  What remains is to show that $j$ is also served at $\mlu{j}{L}{U}$.

Moreover, suppose towards a contradiction that for some player $j\in L$, $\xi(j,O(b))>\mlu{j}{L}{U}$ (If the set $U\setminus L$  is empty $U=L$ it is impossible that some $j \in L$ is charged more than  $\mlu{j}{L}{U}=\xi(j,L)$). Then she could form a coalition with the players in $U\setminus L$, who would enforce the set $S_j$, where $S_j$ is the set guaranteed to exist by Lemma~\ref{lem:a-milder}, to be output i.e. the players $i\in U\setminus L$ could change their bids to $b_i^*$ if $i\in S_j$ and $-1$ if $i\notin S_j$, so that by CS and VP the output is $S_j$.

Since for every $i \in S_j\setminus L$, $\xi(i,S_j)=\mlu{i}{L}{U}$ their utilities remain zero after this manipulation (the same holds trivially for all $i \in U\setminus S_j$), while the utility of player $j$ strictly increases, and thus the coalition is successful indeed. 

Consequently for all $j\in L$, $j\in O(b)$ and $\xi(j,O(b))=\mlu{j}{L}{U}$ and the players in $O(b)\setminus L$ are charged at their minimum payment as well.\qed

\subsubsection*{Proof of Claim \ref{cla:sinkandb}}

Using the induction hypothesis at $L\cup\{k\},U$ we get from part (a) that there is a set $S$, where $L\cup\{k\}\subseteq S\subseteq U$, such that for all $i \in S$, $\xi(i, S)=\mlu{i}{L\cup\{k\}}{U}$. Using now the fact that $k$ is a sink we have that
for all $i \in S \setminus L$, $\m{i}=\mlu{i}{L\cup\{k\}}{U}$. Thus, setting $S_k=S$ we satisfy condition (b) of \lcm  for $k$ at $L,U$. \qed

\subsubsection*{Proof of Claim \ref{claim:harmandU-L}}

If $j$ is not a \emph{sink} if $G[U\setminus L]$, there must be an path starting from $j$, which obviously ends  at a sink $k$ of this graph. Transitivity implies that $j$ \emph{harms} $k$. \qed

\subsubsection*{Proof of Lemma \ref{lem:b-alloc}}

Like in the proof of the  Lemma \ref{lem:a-alloc}, by CS and VP we get that $L'\subseteq O(b^j)\subseteq U'$ and thus from the definition of the bid vector and $\xi^*$, every player in  $U'\setminus (L'\cup\{j\})$ has zero utility. We need to show that $j$ is serviced and charged $\mlu{j}{L'}{U'}$.

Moreover, suppose towards a contradiction that player $j$, is either not serviced or she is charged  an amount greater than $\mlu{j}{L'}{U'}$.
Since we assumed that condition (b) of \lcm  is satisfied for
$j$ at $L',U'$, there must be a set $S_j$ with $L'\subseteq S_j \subseteq U'$, $j \in S_j$ and for all $i \in S_j\setminus L'$, $\xi(i,S_j)=\mlu{i}{L'}{U'}$.
  Obviously if she is serviced at a higher payment then she prefers the set $S_j$ to the current outcome. Also, the same holds if she is not serviced, as we assumed that $b_j >\mlu{j}{L'}{U'}=\xi(j,S)$ and thus her utility would be to a \emph {strictly positive}  if the outcome was $S_j$.  In similar manner like int the proof of Lemma \ref{lem:a-alloc}  she could form a coalition with the players in $U'\setminus (L'\cup\{j\})$ by enforcing $S$ to be output. Again our assumption about the payments of the rest players in $S_j\setminus L'$ implies that their utility remains unchanged thus this coalition is successful.\qed

\subsubsection*{Proof of Claim \ref{claim:bk-alloc}}

(a)  We get that player $k$ is serviced  and charged $\m{k}$ by applying  Lemma \ref{lem:b-alloc} for $k$ with $L'=L$ and $U'=U$.

(b)
Suppose towards a contradiction that $j \in O(b^k)$. The payment of player $k$ would be lower bounded by $\mlu{k}{L\cup\{j\}}{U}$, since $L\cup\{j\}\subseteq O(b^k) \subseteq U$, which contradicts with the fact that $k$ is charged $\m{k}$, which is strictly lower by our assumption that $j$ \emph{harms} $k$. \qed

\subsubsection*{Proof of Claim \ref{claim:bj-alloc}}

(a)  Since $k$ is a sink of $G[U\setminus L]$, we have that for all $i \in U\setminus (L\cup\{j,k\})$, $\mlu{i}{L\cup\{k\}}{U}=\m{i}=b^j_i$. Similarly, we get that $\mlu{j}{L\cup\{k\}}{U}+\epsilon=\m{j}+\epsilon=b^j_j$.

(b) We   apply Lemma \ref{lem:b-alloc} for $j$ with $L'=L\cup\{k\}$ and $U'=U$, since condition (b) is satisfied for $j$ at this pair (induction hypothesis) and  the bid vector satisfies the requirements of this Lemma. ($b^j_i=b^*_i$ for all $i \in L\cup\{k\}$, $b^j_i=-1$ for all $i \notin U$ and using Claim~\ref{claim:bj-alloc} (a).)  As a result, $j\in O(b^j)$ and $\xi(j,O(b^j))=\mlu{j}{L\cup\{k\}}{U}=\m{j}$.

(c) From part (b) we get that  the set $O(b^k)$ satisfies that $L\cup\{j\}\subseteq O(b)\subseteq U$, and thus the payment of player $k$ is lower bounded by $\mlu{k}{L\cup\{j\}}{U}$. Since $j$ \emph{harms} $k$ we get that $\mlu{k}{L\cup\{j\}}{U}>\m{k}$ completing our proof.\qed

\subsubsection*{Proof of Claim \ref{claim:bjk-alloc}}

(a) Assume that player $k$ is serviced at $b^{j,k}$. Notice that by VP and the definition of $\epsilon$, if $k$ is serviced at $b^{j,k}$ then her payment cannot exceed $\m{k}$.  Additionally, notice that the only coordinate $b^{j,k}$ differs from $b^j$ is the bid of player $k$. Thus,  strategyproofness is violated from $b^j$, since from Claim \ref{claim:bj-alloc} the payment of $k$ decreases.

(b)   Suppose that $j$ is not serviced at $b^{j,k}$. Then $\{j,k\}$ can form a successful when true values are  $b^k$ bidding $b^{j,k}$, since from Claim \ref{claim:bk-alloc} the utility of $k$ increases from zero to  $\epsilon$ and the utility of $j$ is kept to zero (she is not serviced at either input).

(c) By VP and CS we get that $L\subseteq O(b) \subseteq U$, thus the payment of every  serviced player $i$ is lower bounded by $\m{i}$. By definition of $b^{i,j}$ and VP of the mechanism we get the equality for every player in $O(b) \setminus (L\cup\{j\})$. Moreover, from the  definition of $\epsilon$, we conclude the same for $j$. \qed

\subsubsection*{Proof of Lemma \ref{lem:char1}}

We will prove our statement with induction on the cardinality of the set $T=\{i \in L\mid b_i \neq b^*_i\}$.

\textbf{Base:} Since $T=\emptyset$, we simply apply Lemma~\ref{lem:a-alloc}.

\textbf{Induction step:} For the induction step we will show that $L\subseteq O(b)\subseteq U$ and for all $i \in L$, $\xi(i,O(b))=\m{i}$. Then it is easy to see that every  $i\in O(b)\setminus L$ must be serviced at $\m{i}$ by VP and the definition of $\xi^*$.

By the construction of the bid vector we have that $O(b)\subseteq U$. Consider now some $j \in T$. Induction hypothesis implies that $j$ is serviced and charged $\m{j}$ at the bid vector $(b^*_j,b_{-j})$. Suppose that $j$ is not serviced at $b$ resulting in zero utility. Since $b_j-\m{j}>0$ she can misreport $b^*_j$ so as to increase her utility to a strictly positive quantity and thus violate the  strategyproofness of the mechanism. As a result, $j$ must be serviced at $b$ and hence by strategyproofness $j$ must be serviced at the same payment.
(since $j$ is an arbitrary element of $T$, the same holds for all $j \in T$).

Now since $j$ is indifferent between the two outcomes group-\-stra\-te\-gy\-pr\-oof\-ness requires the same about the rest players. This is only possible  $i \in L\setminus T$ ($i\in O(b)$ by CS since $b_i=b^*_i$) is charged the same payment in either input i.e. $\xi(i,O(b))=\xi(i,O(b^*_j,b_{-j}))=\m{i}$ (induction hypothesis).
\qed

\subsubsection*{Proof of Lemma~\ref{lem:char2}}
Let $b^0$ be any bid vector satisfying the conditions of Lemma \ref{lem:char1}, which means that for all $i \in L$, $b^0_i>\m{i}$, for all $i \in U\setminus L$, $b^0_i=\m{i}$ and for all $i \notin U$, $b^0_i=-1$. Then we relax the constrains we put on the bids of the players in $U \setminus L$ (the rest players bid always according to $b^0$) and prove that no player becomes serviced at a price strictly less than her minimum payment at $L,U$, by using  induction on $|T|$, where 
$T:=\{i \in U\setminus L\mid b_i\neq b^0_i\}$.

\textbf{Base:} For $T=\emptyset$ we have that $b=b^0$ and the allocation property follows  from Lemma \ref{lem:char1}.

\textbf{Induction step:} Assume that there is some $j \in O(b)$ that is charged less than $\m{j}$. Since by VP $O(b)\subseteq U$, either $j\in L$ which implies that $b^0_j>\m{j}$ or $j \in U\setminus L$ and thus $b^0_j=\m{j}$. In both cases we have that
\begin{equation}\label{eq:char2-1}
j\in O(b)\text{ and }\xi(j,O(b))<\m{j}\leq b^0_j.
\end{equation}

We will prove that there exists \emph{at least one} successful coalition, contradicting the assumed group-strategyproofness of the mechanism.
 The key of our proof is the definition of the following special subset of $T$
\begin{equation}\label{eq:char2-2}
R:=\{ i\in T \mid i \in O(b)\text{ and  } \xi(i,O(b))>b^0_i\}
\end{equation}
(notice that $j\notin R$).

This set is a special subset of $T$ that can be interpreted as follows: If the true valuations are given by $b^0$ and the players in $T$ bid according to $b$, then the set $R$ represents the players that their utility becomes negative, and since every $i \in R$ had zero utility at $b^0$ (Lemma \ref{lem:char1}), this gives us the set of \emph{liars that sacrifised their utilities}. Notice that the set $R$ not being empty  renders the manipulation we considered \emph{unsuccessful}. For proving the induction step we consider two cases regarding $R$.

\textbf{Case 1:} If $R \subset T$, we construct the  bid vector $b'$, where for every $i\in R$, $b'_i=b_i$ and for every $i\notin R$, $b'_i=b^0_i$.  We complete the proof of this case by showing that  $(T\setminus R)\cup\{j\}$  form a coalition when true values are $b'$ bidding $b$. 

First, we  prove that  in the truthful scenario  every player $i\in T\setminus R$ has zero utility and non-negative after the misreporting (profile $b$). By using induction hypothesis at $b'$ (differs with $b^0$ in $|R|<|T|$ coordinates) we get that for every $i \in T\setminus R$ such that $i \in O(b')$ (the rest players  have obviously zero utility), it holds that  $\m{i}\leq\xi(i,O(b'))\leq b'_i$, where the last inequality follows from VP. Since $i \in T\setminus R$, we get that 
$b'_i=b^0_i=\m{i}$ and thus $b'_i=\xi(i,O(b'))$. Now consider the outcome after the misreporting. Either $i \notin O(b)$ and hence she has zero utility after the misreporting or $i \in O(b)$ and her utility becomes  $b'_i-\xi(i,O(b))= b^0_i-\xi(i,O(b))\geq 0$, where the last inequality follows from Equation \ref{eq:char2-2} as $i\in T\setminus R$.

Second, we prove that player $j$ strictly increases her utility. Considering the truthful scenario there are two cases for $j$: She is serviced at $b'$ and charged $\xi(j,O(b'))\geq\m{j}$ from the induction hypothesis. From  Equation \ref{eq:char2-1} we get that  $j$ is serviced after the misreporting and charged $\xi(i,O(b))<\m{j}$, and thus her payment decreases, while she is still serviced.  Now if  $j$ is not serviced at $b'$, then she has zero utility. From Equation \ref{eq:char2-1} we get that $b^0_j>\xi(j,O(b))$ and since $j \notin R$ it follows that $b^0_j=b'_j$. As a result, her utility increases to   $b'_j-\xi(j,O(b))>0$.

\textbf{Case 2:} Otherwise if $R = T$, we construct the bid vector $b''$, where every $i \in R$,  $b''_i=\xi(i,O(b))$ and every $i \notin R$, $b''_i=b_i$. 
Notice that for all $i \in R$, $b_i\geq b''_i$ by VP at $b$ since $R\subseteq O(b)$. Moreover, for every $i \notin R$, $b''_i=b^0_i$, since $R=T$.  

\begin{claim}\label{claim:bdp}
If $T=R$ and there is some $j$ that satisfies Equation \ref{eq:char2-1}, then for the bid vector $b''$, it holds that 
$R\cup\{j\}\subseteq O(b'')$ and for all $i \in R\cup\{j\}$, 
$\xi(i,O(b''))=\xi(i,O(b))$.
\end{claim}

\begin{proof}[of Claim~\ref{claim:bdp}]

 For all $S\subseteq R$, we define the bid vector $b^S$, where for all $i \in S$, $b^S_i=b''_i$ and for all $i \notin S$,
$b^S_i=b_i$. Notice that since we assumed that $T=R$, it holds that $b^S_i=b^0_i$   for all $i \notin R$ and all $S\subseteq R$.  We show that for all $S\subseteq R$, $R\cup\{j\}\subseteq O(b^S)$ and for all $i \in R\cup\{j\}$, 
$\xi(i,O(b^S))=\xi(i,O(b))$, with induction on $|S|$.
We will refer to this induction as \emph{second} induction to discriminate it from the \emph{first}. After proving this statement, we can set $S=R$ ($b^R=b''$) and complete the proof.

\textbf{Base (second):} If $S=\emptyset$, i.e. $b^\emptyset =b$, then every condition holds trivially.

\textbf{Induction step (second):}
If for some $i \in S$, it holds that $b''_i=b_i$, then the induction step follows trivially from induction hypothesis as $b^S=(b_i,b^S_{-i})=b^{S\setminus \{i\}}$.

Now assume that for all $i \in S$, $b''_i< b_i$ and consider some $i \in S$.
Player $i$ is serviced at $(b_i,b^S_{-i})$  and charged 
$\xi(i,O(b))$ (induction hypothesis). Strategyproofness implies that by lowering her bid up to her payment, which gives us the bid vector $b^S$, if she is  serviced, then her payment must be the same, i.e.,
\begin{equation}\label{eq:char2-4}
\text{for all }i \in S\cap O(b^S), \text{ } \xi(i,O(b^S))=\xi(i,O(b)).
\end{equation}

This Equation implies that if the true values are given by $b^S$, then every player in $S$ has zero utility. Moreover, by misreporting  $b$ they are charged an amount equal to their true value and thus their utilities is kept at zero. Group-strategyproofness implies that no player in $\A\setminus S$ has incentives for this misreporting. Notice that $b^S_j=b^0_j$ since  $j \notin R=T$.   Moreover, from Equation \ref{eq:char2-1} we have that $b^0_j\geq \m{j}>\xi(j,O(b))$ and thus her utility after the manipulation we described is given by 
$b^S_j-\xi(j,O(b))>0$. In order that this  amount is not bigger than her utility in the truthful scenario ($b^S$), it holds that 
\begin{equation}\label{eq:char2-3}
j \in O(b^S)\text{ and }\xi(j,O(b^S)\leq\xi(j,O(b)),
\end{equation}
 which together with Equation \ref{eq:char2-1} and the fact that $b^S_j=b^0_j$ implies that
\begin{equation}\label{eq:char2-5}
\xi(j,O(b^S))<\m{j}\leq b^S_j.
\end{equation}

Using the latter equation we show that $R\subseteq O(b^S)$. Assume that some $i\in R$ is not serviced at $b^S$.  First, we show that $i$ is not serviced at
$(b^0_i,b^S_{-i})$.  There are two cases for $i$: Either $b^S_i=b''_i$ if $i \in S$ or $b^S_i=b_i\geq b''_i$  if $i \in R\setminus S$. In any case $b^S_i\geq b''_i>b^0_i$, where the second inequality follows from  the definition of $R$ (Equation \ref{eq:char2-2}) since $b''_i=\xi(i,O(b))$. Strategyproofness implies that if the true values are $b^S$ and $i$ reports $b^0_i$ (the bid vector becomes $(b^0_i,b^S_{-i}))$, then she is not serviced as otherwise  her payment cannot exceed $b^0_i$ by VP and thus her utility increases to  $b^S_i-b^0_i>0$. 

Now we show that if  $i$ is not serviced at $(b^0_i,b^S_{-i})$, then $\{i,j\}$ from a successful coalition when true values are $(b^0_i,b^S_{-i})$ bidding $b^S$.  The utility of player $i$ is kept to zero, as she is excluded in both inputs from the outcome. Notice that by the \emph{first} induction hypothesis at the bid vector 
$(b^0_i,b^S_{-i})$ (differs with $b^0$ in $|R|-1=|T|-1$ coordinates) there are two cases for $j$: Either $j$ is serviced and  charged $\xi(j,O((b^0_i,b^S_{-i})))\geq\m{j}$ and thus her payment decreases to  $\xi(j,O(b^S))<\m{j}$ (from Equation  \ref{eq:char2-5}) or
 she is not serviced and her utility becomes $b^S_j-\xi(j,O(b^S))>0$ (from Equation  \ref{eq:char2-5}).

 As a result, it holds that   $S\subseteq R\subseteq O(b^S)$ and together with \ref{eq:char2-3} we get that for every $i \in S$,  $\xi(i,O(b^S))=\xi(i,O(b))$. Therefore, the players in $S$ are indifferent between the two inputs, i.e regardless of which of the two we consider as true values, they can misreport the other without losing utility. Group-strategyproofness implies that the other players are indifferent as well, which is only possible if for all $i \in R\setminus S$, $\xi(i,O(b^S))=\xi(i,O(b))$. Moreover, from Equation \ref{eq:char2-3} we get that $j \in O(b^S)$ and for similar reasons it holds that $\xi(j,O(b^S))=\xi(j,O(b))$. 
\end{proof}

By VP we have that $O(b'')\subseteq U$ (for every $i \in \A\setminus U$, it holds that $i \notin R$ and thus $b''_i=b^0_i=-1$). Furthermore, from Claim \ref{claim:bdp} and Equation \ref{eq:char2-1} we get that  $\xi(j,O(b''))<\m{j}$.  Therefore, the definition of $\xi^*$  implies that  $L\not\subseteq O(b'')$.

Assume that the true values are given by $b''$. 
Every player $k \in (L\setminus O(b''))$ has obviously  zero utility. 
From Claim~\ref{claim:bdp} we deduce the same for every $i \in R$.
 We complete our proof by  showing  that $R\cup (L\setminus O(b''))$ form a successful coalition   bidding $b^0$.

First, we prove that the utility of every $i \in R$ remains non-negative after the misreporting. Consider some $i \in R$ such that $i \in O(b^0)$ (obviously the rest players have zero utility). From Lemma \ref{lem:char1} it holds that $\xi(i,O(b^0))=\m{i}$ and since $\m{i}=b^0_i$ ($i \in U\setminus L$) and $b^0_i<\xi(i,O(b))=b''_i$ (Equation \ref{eq:char2-2}), we conclude that  her utility increases to $b''_i-\xi(i,O(b^0))>0$. 

Second, we show that the utility of every $k \in L\setminus O(b'')$ increases, as it becomes strictly positive.  Since $k \notin R$  we get that  $b''_k=b^0_k$. Moreover, since $k \in L$ we have  $b^0_k>\m{k}$ and 
from Lemma \ref{lem:char1} we get that $k \in O(b^0)$ and  $\xi(k,O(b^0))=\m{k}$. As a result, her   utility becomes $b''_k-\xi(k,O(b^0))>0$.
\qed

\subsubsection*{Proof of Claim \ref{claim:bc-alloc}}

(a) We can apply Lemma \ref{lem:char2} as every condition is satisfied.

(b) VP and CS implies $C\subseteq O(b^C) \subseteq L\cup C$. $C\subset O(b^C)$ follows from (a) and our assumption that $\xi(j,C)<\m{j}$.

(c) From definition of $\epsilon$ we have that for all $i \in O(b^C) \setminus C$,  $\xi(i,O(b^C))\leq\m{i}$. Together with (a) we get the equality.
\qed

\subsubsection*{Proof of Lemma \ref{lem:char3}}

Notice that Lemma \ref{lem:char1} implies the same allocation property in the special case, where every $i \in \A \setminus U$ has bidden $-1$.
In order to show this Lemma we will use the same induction, i.e. we show that $L\subseteq O(b)$ and for  all $i \in O(b)$, $\xi(i,O(b))=\m{i}$ using induction on $m=|\{i\in L\mid b_i\neq b^*_i\}|$.

\textbf{Base:} First, by CS we have that $L\subseteq U$. Next, we show that not player in $\A\setminus U$ is serviced. Suppose that the set $R=O(b)\setminus U\neq\emptyset$.  We will show that VP is violated. By the
definition of the bid vectors we consider it holds that there
is some $i \in R$, such that $b_i<\mlu{i}{L}{U\cup R}\leq \xi(i,O(b))$
(since $L\subseteq O(b)\subseteq U\cup R$) and thus we reach a contradiction.

As a result, it holds  that $L \subseteq O(b) \subseteq U$. Thus, for every
$i \in O(b)$, we have that $\xi(i, O(b)) \geq \m{i})$. By VP
we get the equality for every $i \in O(b) \setminus L$. Now assume
that the previous inequality is strict for some $j \in L$. We
show that the players in $(\A \setminus U )\cup \{j\}$ can form a successful
coalition, where every $i \notin U$ announces $-1$.

Trivially the utilities of these players are kept to zero after
the misreporting and applying Lemma \ref{lem:char1} we get that $j$ is serviced
and charged $\m{j}< \xi(j, O(b))$.

\textbf{Induction Step}: We show that $L \subseteq O(b)$ and for all $i \in
L$, $\xi(i, O(b)) = \m{i}$ exactly like in the proof of Lemma
\ref{lem:char1}. Now since $L \subseteq O(b)$ in similar way we show that $U \subseteq O(b)$.
Thus, this restriction together with the definition of the bid
vector, implies that for every $i \in O(b) \setminus L$, $\xi(i, O(b)) =
\m{i}$.\qed

\section{Missing Proofs of Section 6}

\subsubsection*{Proof of Lemma \ref{lem:maxu}}

Assume towards a contradiction that there exist two distinct sets $U_1,U_2$ none of which is a subset of the other that both satisfy this property. Then we could construct the set $U_1\cup U_2$ which also satisfies the same property and is a proper superset of both $U_1$ and $U_2$ reaching a contradiction. Indeed for all $i\in U_1\setminus L$ we have $b_i\geq\mlu{i}{L}{U_1}\geq \mlu{i}{L}{U_1\cup U_2}$ as the minimum payment of player $i$ can only decrease as the set of outcomes, over which the minimum in the definition of $\xi^*$ is taken, becomes larger. A similar inequality holds for the players in $i\in U_2\setminus L$ completing the proof.

Now we show that if  $U$ is the maximal set with this property, then $L,U$ satisfy property 3. of stablity. Consider some non-empty $R \subseteq  \A \setminus U$. Assume that for all $i \in  R$, $b_i  \geq  \mlu{i}{L}{U\cup R}$. From Lemma \ref{lem:xicompare} we have 
that for all $i \in U \setminus L$, $b_i \geq \mlu{i}{L}{ U\cup R}$ and thus we reach contradiction by the maximality of $U$.
\qed

\subsubsection*{Proof of Lemma \ref{lem:uncov2}}

We will first show that if $S\not\subseteq U$ then there exists some $i\in S$ we have $\mlu{i}{L}{U\cup S}>\xi(i,S)$. Since $L,U$ is stable (property 3.) there exists some $i\in S\setminus U$ such that $b_i<\mlu{i}{L}{U\cup S}$ and since from the initial assumption $\xi(i,S)\leq b_i$ we get that $\xi(i,S)<\mlu{i}{L}{U\cup S}$ 

The rest of the proof is the same for both cases (note that in what follows, if $S\subseteq U$ then $S\cup U=U$). Note that $S\subset U\cup S$, since $L\nsubseteq S$. Applying part (c) of \lcm    we get that there exists a non-empty $T\subseteq L\setminus S$ such that for all $j\in T$ we have $\xi(j,S\cup T)=\mlu{j}{L}{U\cup S}\leq \m{j}$, where the last inequality is by the definition of $\xi^*$, since the minimum cannot decrease as the set of outcomes over which it is taken becomes larger.  From property 1. of stability and since $T\subseteq L$, for all $j\in T$ we have $\m{j}<b_i$. Consequently for all $j\in T$ we have $\xi(j,S\cup T)<b_j$.
\qed

\subsubsection*{Proof of Lemma \ref{lem:cov3}}

Suppose towards a contradiction that $L,U$ is stable at $b$ and that there exists some $T\subseteq \A\setminus S$, such that for all $i\in T$, $b_i\geq \xi(i,S\cup T)$ and that for at least one player the inequality is strict.

Since $L\subseteq S\cup T\subseteq U\cup T$ from the definition of $\xi^*$ we get that for all $i\in  T$  we have  $\mlu{i}{L}{U\cup T}\leq \mlu{i}{L}{U}\leq \xi(i,S\cup T)$.

If $T\subseteq U$, then there exists some $i\in T$ such that $b_i>\xi(i,S\cup T)\geq \m{i}$. Since  $T\subseteq \A\setminus L$ this contradicts with stability (condition 2.) of $\bv$ at $L,U$.

If $T\not\subseteq U$, then $T\setminus U$ is not empty and for all $i  \in T\setminus U$ we have $b_i\geq \xi(i,S\cup T)\geq \mlu{i}{L}{U\cup T}$,  which contradicts with stability (condition 3.) of $b$ at $L,U$.
\qed

\subsubsection*{Proof of Lemma \ref{lem:unique}}

Assume first that $L_1\neq L_2$ and without loss of generality that $L_2\not\subseteq L_1$. Consequently there exists some $j\in L_2\setminus L_1$.

By \lcm  (part (a)) there exists a set $S_2$ such that
$L_2\subseteq S_2\subseteq U_2$ and $\xi(i,S_2)=\mlu{i}{L_2}{U_2}$ for all $i\in S_2$.  Notice that for all $j\in S_2$, $b_j \geq \xi(j, S_2)$, where strict inequality holds only if $j  \in L_2$.
The idea is to apply Lemma~\ref{lem:uncov2} and get that there exists a non-empty $T\subseteq L\setminus S_2$, such that for all $i\in T$ we have $b_i>\xi(i,S_2\cup T)$. We will then apply Lemma~\ref{lem:cov3} to show that $L_2,U_2$ is not stable at $\bv$ which contradicts our intial assumption.

It only remains to show that we can apply Lemma~\ref{lem:uncov2}. If $S_2\not\subseteq U_2$ this is immediate. Suppose that $S_2\subseteq U_1$.  If $j\in L_2\setminus L_1$ then also $j \in U_1 \setminus L_1$ thus  from stability (condition 2.) of $L_1,U_1$, we get that  $b_j=\mlu{j}{L_1}{U_1}$ and from stability (condition 1.) of $L_2,U_2$ and since $j\in L_2$, we get that $b_j>\xi(j, S_2)$. Therefore, we get that $\xi(j,S_2)<\mlu{j}{L_1}{U_1}$ and consequently $S_2$ satisfies the requirements of Lemma~\ref{lem:uncov2}.

We showed that $L_1=L_2=L$. Suppose towards a contradiction that $U_1\neq U_2$. From Lemma \ref{lem:maxu} there exists a unique maximal set $U$ such that $b_i\geq \m{i}$ for all $i\in U$. Consequently $U_1,U_2$ are subsets of $U$ and at least one of them, say $U_1$ a proper subset of $U$. Then the players in $U\setminus U_1$ contradict stability (condition 3.).
\qed

\subsubsection*{Proof of Claim \ref{claim:inclusion*}}

(a) Using the definition of $L^*$ we get that $L^*\cup \{i\}$ contains all the players who have bidden higher than any payment of the mechanism in $(b_i^*,b_{-i})$.
Since $L_i$ is stable at $(b_i^*,b_{-i})$, each one of these players, who have bidden strictly higher than any payment of the mechanism, should be serviced (any value higher than any payment satisfies the definition of CS for Fencing Mechanisms), and if they are serviced they obviously have strictly positive utility, and thus $L^*\cup \{i\}\subseteq L_i$.

(b) 
The idea is to show that for all $j \in (U^*\cup U_i)\setminus L^*$, $b_j\geq \mlu{j}{L^*}{U^*\cup U_i}$ and since we defined $U^*$ be the maximal set with this property, we get that $U_i\subseteq U^*$.

From definition of $U^*$ we have that for all $j \in U^*\setminus L^*$,  
$b_j\geq \mlu{j}{L^*}{U^*} \geq \mlu{j}{L^*}{U^*\cup U_i}$, where the last inequality follows from Lemma \ref{lem:xicompare}.

Now consider some $j \in U_i\setminus U^*$ ($j\neq i$). Since $L_i,U_i$ is stable at $(b^*_i,b_{-i})$ we get that  $b_j\geq \mlu{j}{L_i}{U_i}\geq  \mlu{j}{L^*}{U_i\cup U^*}$ by applying Lemma \ref{lem:xicompare}, since from (a) $L^*\subset L_i$.

(c)  From (a),(b) and Lemma \ref{lem:xicompare} we have that $\mlu{j}{L_i}{U_i}\geq \mlu{j}{L^*}{U^*}$. As we defined $L_i,U_i$ to be the stable pair at $(b_i^*,b_{-i})$ and $j\in L_i$, for $j\neq i$ we have $b_j>\mlu{j}{L_i}{U_i}\geq \mlu{j}{L^*}{U^*}$ and as $j\notin L^*$, we have $j\in W$. \qed

\subsubsection*{Proof of Claim \ref{claim:subset}}

(a) The pair $L_i,U_i$ is stable at $(b_i^*,b_{-i})$ (by definition of $L_i,U_i$) but it is not stable at $b$ (from our assumption that there exists no stable pair at $b$). Since the two bid vectors differ only on the $i$-th coordinate we deduce that stability (condition 1. as $i\in L_i$) is not satisfied by the $i$-th coordinate of $b$, thus $b_i\leq \mlu{i}{L_i}{U_i}$.

(b) From Claim~\ref{claim:inclusion*} (a) we already have that $L_i\supseteq L^*\cup \{i\}$. Suppose towards a contradiction that $L_i=L^*\cup \{i\}$. From \lcm  (condition (b)) we get that there exists some $S_i$, where $L^*\subseteq S_i\subseteq U^*$ and $i\in S_i$ such that for all $j\in S_i\setminus L^*$ we have $\xi(j,S_i)=\mlu{j}{L_i}{U_i}$.

The idea is to show that $L_i\subseteq S_i \subseteq U_i$, which implies that $\mlu{i}{L_i}{U_i}\leq \xi(i,S_i)=\mlu{i}{L^*}{U^*}$. Then considering also that $\mlu{i}{L^*}{U^*}<b_i$, because $i \in W$,  we get $\mlu{i}{L_i}{U_i}<b_i$, contradicting the inequality we showed in (a).

It only remains to show that $L_i\subseteq S_i \subseteq U_i$. Notice first that $L_i\subseteq S_i$ since $L_i=L^* \cup\{ i \}$. Moreover, since for all $j \in S_i \setminus L_i$,  it holds that $b_j \geq\xi(j, S_i)$ and  the bids of these players are the same at $(b^*_i,b_{-i})$, stability of
$L_i,U_i$ (condition 3) implies that $S_i\subseteq U_i$, because otherwise $S_i\setminus U_i\neq \emptyset $ and its elements would violate condition 3. of  stability.
\qed

\subsubsection*{Proof of Claim \ref{claim:notgsp}}

(a) The pair $L_i,U_i$ is stable at bid vector $(b_i^*,b_j^*,\bv_{-\{i,j\}})$, since  $L_i,U_i$ is stable at $(b_i,b_{-i})$ and since raising the bid of player $j$ ($j\in L_i$ from our initial assumption) to $b_j^*$ does not effect its stability.

(b) From stability of $L_i,U_i$ at $(b_i^*,b_{-i})$ we get that $b_j>\mlu{j}{L_i}{U_i}$, while from part (a)  of Claim \ref{claim:subset} we  get that $\mlu{j}{L_j}{U_j}\geq b_j$. Consequently $\mlu{j}{L_j}{U_j}>\mlu{j}{L_i}{U_i}$.

 Supposing towards a contradiction that $i\in L_j$ we also get in a similar way as before that the pair $L_j,U_j$ is stable at $(b_i^*,b_j^*,\bv_{-\{i,j\}})$. By Lemma \ref{lem:unique} these two pairs coincide, which is a contradiction because we just showed that the payment of $j$ is different.

(c) We will show that if there exists some $k\in L_j\setminus L_i$ then the mechanism is not GSP at inputs, where a stable pair exists. Observe that $k\neq i,j$.

Consider first the vector $b^{i,j}:=(b_i^*,b_j^*,\bv_{-\{i,j\}})$, which stable pair is $L_i,U_i$ from part (a). As $k\notin L_i$ either $k$ is not serviced, or if $k$ is serviced, then her utility is zero. As for $j$ she is serviced with payment $\mlu{j}{L_i}{U_i}<\mlu{j}{L_j}{U_j}=\xi(j,O(b_j^*,b_{-j}))$ (from (a)).

Consider then $b^j:=(b_j^*,\bv_{-j})$, where $L_j,U_j$ is stable. As $k\in L_j$ and $k$ is serviced and $b_k>\xi(k,O(b_j^*,\bv_{-j}))$.

Resuming player  $j$ strictly prefers $O(b^{i,j})$ to $O(b^j)$, while for $k$ the situation is exactly the opposite. The idea is to construct a bid vector $b'$ where $i$ has zero utility and either $\{i,j\}$ or $\{i,k\}$ is a successful coalition. Let $b'=(\xi(i,O(b^{i,j})),b_j^*,b_{-\{i,j\}})$ and notice that induction hypothesis implies that there is a stable pair at $b'$.  Moreover, observe that the three bid vectors differ only on the bid of player $i$ and consequently from strategyproofness, at every input $i$ is served, she is charged $\xi(i,O(b^{i,j}))$.
This implies that if $i$ is served at $b'$ she is charged an amount equal to her bid. Moreover, $i$ is serviced at $b^j$ in the degenerate case where $b'=b^j$ as otherwise $VP$ would imply that
$\xi(j, O(b^j))\leq b^j_j< b^{i,j}_j=\xi(j, O(b^{i,j}))$. In every case we have that $i$ has zero utility in
at $b'$ and $b^j$.

Observe first that  $k$ must be serviced at $b'$ and charged $\xi(k,O(b'))\leq\xi(k,O(b^{j}))$   as otherwise $\{i,k\}$ would have been able to form a successful coalition when the true values are $b'$ bidding $\bv^j$. Similarly we can show that $\xi(j,O(b'))\leq\xi(j,O(b^{i,j}))$, excluding the degenerate case $b'=b^j$, because otherwise $\{i,j\}$ would have been able to form a successful coalition when the true values are $b'$ bidding $b^{i,j}$ (the utility of $i$ is kept to zero as in the truthful scenario).

As a result players $j$ and $k$  strictly prefer $O(b')$ to $O(b^j)$ and $O(b^{i,j})$ respectively.
If $i\in O(b')$ then $\{i,k\}$ when true utilities are given by $b^{i,j}$ can form a successful coalition bidding $b'$ strictly increasing the utility of $k$, while keeping the utility of $i$ constant, since she is still served with he same payment. If $i\notin O(b')$ then we deduce that $\{i,j\}$  when the true utilities are given by $b^j$ can form successful coalition bidding $b'$ strictly increasing the utility of $j$, while keeping the utility of $i$ to zero.\qed

\section{Missing Proofs of Section 7}
\subsubsection*{Proof of Proposition \ref{prop:neg-neu}}

Let $L=S\setminus \{i,j\}$ and $U=S$.

(a) From condition (a) of \lcmc we get that at either $S\setminus \{i\}$ or $S$,
every player is charged the minimum payment at $L,U$. Since by our assumption, this is not true for $S$, as $\xi(j,S\setminus\{i\})<\xi(i,S)$, it follows that every other player $k \in S\setminus \{i,k\}$, $\xi(k,S\setminus\{i\})=\m{k}$ $\Rightarrow$ $\xi(k,S\setminus \{i\})\leq\xi(k,S)$.

(b) Notice that $j$ belongs only to
$S\setminus \{i\}$ and $S$. Moreover, our assumption implies that $\m{j}=\xi(j,S\setminus \{i\})<\xi(j,S)$.
From condition (b) of \lcmc there must be a set $S_i$ with $L\subseteq S_i\subseteq U$, such that $i \in S$ and for all $k \in S_i\setminus L$,
$\xi(k,S_i)=\m{k}$. By our assumption, it is impossible that $S_i=S$ and since the only remaining set that contains $i$ is $S\setminus \{j\}$, we get that
$\xi(i, S\setminus \{j\})=\m{i}$ $\Rightarrow$ $\xi(i, S\setminus\{j\})\leq\xi(i,S)$.

(c) Let $L'=S\setminus \{j\}$. Now  $j$ is contained only at $S$, it follows that $\xi(j,S)=\mlu{j}{L'}{U}>\xi(j,S\setminus \{i\})$.
Since $L'\setminus  S\setminus\{i\}=\{i\}$, condition (c) of \lcm  implies that
$\mlu{i}{L\cup\{i\}}{U}=\m{i}=\xi(i,S)$ $\Rightarrow$
$\xi(i,S)\leq\xi(i,S\setminus \{j\})$.\qed

\subsubsection*{Proof of Theorem \ref{theo:low}}

 Assume by contradiction that there is a
$\alpha$-budget balanced  cost sharing scheme that satisfies \lcm  for $C$ where
$1\geq \alpha>\frac{1}{x}$. This implies the following relations

\begin{eqnarray}
 \frac{1}{x} C(S) <\sum_{i\in S}\xi(i,S) &\leq& C(S). \label{eq:bb}
\end{eqnarray}

 Equation \ref{eq:bb} (right part) with $S=\{1,2\}$ and $S=\{1,3\}$ imply that $\xi(2,\{1,2\})\leq 1$ and $\xi(3,\{1,3\})\leq 1$. We apply  Proposition \ref{prop:neg-neu} (c) and we get that either $\xi(2,\{1,2,3\})\leq 1$ or $\xi(3,\{1,2,3\})\leq 1$ (or both). W.l.o.g we assume that $\xi(2,\{1,2,3\})\leq 1$.

Suppose that $\xi(2,\{2,3\})\leq 1$ $\Rightarrow$ $\xi(2,\{2,3\})<\xi(2,\{2\})$ (Equation \ref{eq:bb} (left part) with $S=\{2\}$: $\xi(2,\{2\})>1$). We apply the contra-positive implication of Proposition \ref{prop:neg-neu} (b) and deduce that $\xi(3,\{2,3\})\leq \xi(3,\{3\})$ $\Rightarrow \xi(3,\{2,3\})\leq x$ (Equation \ref{eq:bb} (right part) with $S=\{3\}$: $\xi(3,\{3\})\leq x)$).  This contradicts Equation \ref{eq:bb} (left hand) with $S=\{2,3\}$, since  $\xi(2,\{2,3\})+\xi(3,\{2,3\}) \leq  x+1$.

Suppose that $\xi(2,\{2,3\})>1$ $\Rightarrow$ $\xi(2,\{2,3\})>\xi(2,\{1,2,3\})$ (from our assumption). Notice that the contra-positive  implication of Proposition \ref{prop:neg-neu} (a)  implies that $\xi(3,\{1,2,3\})\leq\xi(3,\{2,3\})$ $\Rightarrow$ $\xi(3,\{1,2,3\})\leq x^2+x-1$ (from Equation \ref{eq:bb} (right part) with $S=\{2,3\}$ it follows that $\xi(3,\{2,3\})\leq x^2+x-\xi(2,\{2,3\})\leq x^2+x-1$). Using now the contra-positive implication of Proposition \ref{prop:neg-neu} (b) we get that $\xi(1,\{1,2,3\})\leq\xi(1,\{1,3\})$
$\Rightarrow$ $\xi(1,\{1,2,3\})\leq 1$ (Equation \ref{eq:bb} (right part) at $S=\{1,3\}$: $\xi(1,\{1,3\})\leq 1$). This contradicts Equation \ref{eq:bb} (left part) with $S=\{1,2,3\}$, as
$\xi(1,\{1,2,3\})+\xi(2,\{1,2,3\})+\xi(3,\{1,2,3\})\leq 1+1+x^2+x-1=x^2+x+1$.
\qed

\section{Missing Proofs of Section 8}
\subsubsection*{Proof of Theorem \ref{theo:compl}}

Consider the following process, which takes as input a bid vector $b$ and
a set $L$.

\begin{algorithmic}
\REPEAT
\STATE $U\leftarrow L\cup \{i\in U\setminus L\mid b_i\geq\m{i}\}$
\UNTIL{For all $i \in U\setminus L$, $b_i\geq\m{i}$}
\end{algorithmic}

We will prove that if we  feed this process with the lower set $L$ of the stable pair at $b$, then the outcome is the  upper set of the stable pair.

Obviously by the definition of the process, for its final set $U$  it holds that for all $i \in U$, $b_i\geq\m{i}$. First, we show that $U$ is the maximal set with this property.

Assume that there is some $U'$ with $U'\not\subseteq U$, that satisfies this property, then we have that for all $i \in (U\cup U') \setminus L$, $b_i\geq\mlu{i}{L}{U\cup U'}$. Thus, it is impossible that the players in $U \cup U'$ are removed at any step of the previous process, contradicting our assumption that $U\neq U'\cup U$ is the outcome of this process.

Now consider now the upper set $U''$ of the stable pair at $b$. Since $U''$ satisfies this property it  follows that $U''\subseteq U$. Notice that if  $U''\subset U\neq \emptyset$, then the elements of $U\setminus U''$ would violate stability of $L,U''$. Thus, $U''=U$ and consequently this process outputs the upper set of the stable pair.

As a result, given an outcome $S$ of a GSP mechanism we can compute $L:=\{i\in S\mid b_i>\xi(i,S)\}$ and use this process to find the upper set $U$. The time-complexity of this algorithm is polynomial in the number of players assuming that we have polynomial-time  access  to every $\m{i}$ for all $L\subseteq U\subseteq \A$ and all $i \in U$. \qed
\end{document}